\pdfoutput=1
\RequirePackage{ifpdf}
\ifpdf 
\documentclass[pdftex]{sigma}
\else
\documentclass{sigma}
\fi

\numberwithin{equation}{section}

\newtheorem{thm}{Theorem}[section]

\newtheorem{prop}[thm]{Proposition}

{\theoremstyle{definition}
\newtheorem{defn}[thm]{Definition}
\newtheorem{Remark}[thm]{Remark}
\newtheorem{Assumption}[thm]{Assumption ${\mathcal D}$-pb}
\newtheorem{AssumptionW}[thm]{Assumption ${\mathcal D}$-pbw}
}

\begin{document}


\newcommand{\arXivNumber}{1509.03822}

\renewcommand{\PaperNumber}{078}

\FirstPageHeading

\ShortArticleName{${\mathcal D}$-Pseudo-Bosons, Complex Hermite Polynomials, and Integral Quantization}

\ArticleName{$\boldsymbol{\mathcal D}$-Pseudo-Bosons, Complex Hermite Polynomials,\\ and Integral Quantization}

\Author{S.~Twareque ALI~$^{\dag^1}$, Fabio BAGARELLO~$^{\dag^2}$ and Jean Pierre GAZEAU~$^{\dag^3\dag^4}$}

\AuthorNameForHeading{S.T.~Ali, F.~Bagarello and J.P.~Gazeau}

\Address{$^{\dag^1}$~Department of Mathematics and Statistics, Concordia University,\\
\hphantom{$^{\dag^1}$}~Montr\'eal, Qu\'ebec, Canada H3G 1M8}
\EmailDD{\href{mailto:twareque.ali@concordia.ca}{twareque.ali@concordia.ca}}

\Address{$^{\dag^2}$~Dipartimento di Energia, ingegneria dell'Informazione e modelli Matematici,\\
\hphantom{$^{\dag^2}$}~Scuola Politecnica, Universit\`a di Palermo, I-90128  Palermo,  and INFN, Torino, Italy}
\EmailDD{\href{mailto:fabio.bagarello@unipa.it}{fabio.bagarello@unipa.it}}
\URLaddressDD{\url{http://www.unipa.it/fabio.bagarello}}

\Address{$^{\dag^3}$~APC, UMR 7164, Univ Paris  Diderot, Sorbonne Paris-Cit\'e, 75205 Paris, France}
\EmailDD{\href{mailto:gazeau@apc.univ-paris7.fr}{gazeau@apc.univ-paris7.fr}}
\Address{$^{\dag^4}$~Centro Brasileiro de Pesquisas F\'isicas,  Rio de Janeiro, 22290-180 Rio de Janeiro, Brazil}

\ArticleDates{Received March 28, 2015, in f\/inal form September 21, 2015; Published online October 01, 2015}

\Abstract{The ${\mathcal D}$-pseudo-boson formalism is illustrated with two examples. The f\/irst one involves deformed complex Hermite polynomials built using f\/inite-dimensional irreducible representations of the group ${\rm GL}(2,{\mathbb C})$ of invertible $2 \times 2$ matrices with complex entries. It reveals interesting aspects of these representations. The second example is based on a pseudo-bosonic generalization of  operator-valued functions of a complex variable which resolves the identity. We show that such a generalization allows one to obtain a quantum pseudo-bosonic version of the complex plane viewed as the canonical phase space and to understand functions of the pseudo-bosonic operators as the quantized versions of functions of a complex variable.}


\Classification{81Q12; 47C05; 81S05; 81R30; 33C45}


\section{Introduction}

Two new illustrations of the ${\mathcal D}$-pseudo-boson (${\mathcal D}$-pb) formalism~\cite{bagbook} are presented in this paper. Both display original and non-trivial results.  The f\/irst one involves a family of biorthogonal polynomials, named deformed complex Hermite polynomials, various properties of which have been worked out during the past years (see for instance~\cite{alimourshah,bsali} and references therein). Their construction involves a {\em deformation} of  the well-known bivariate complex Hermite polyno\-mials~\cite{ghanmi15,Gha,Ism1,ismail15,Ism2,Ism3,Ism4} using f\/inite-dimensional irreducible representations of the group ${\rm GL}(2,{\mathbb C})$ of invertible $2 \times 2$ matrices with complex entries and reveals interesting aspects of these representations. The second example introduces families of vectors and operators in the underlying Hilbert space built in the same same way as standard coherent states, as orbits in the Hilbert space of the projective Weyl--Heisenberg group. An appealing consequence of this construction is the resolution of the identity satisf\/ied by these families,  possibly on a dense subspace.  Hence,  it becomes possible to proceed with  integral quantizations~\cite{bergaz14,gazhell15}, which then yield  the correct pseudo-bosonic commutation rules. This unorthodox path to the quantum world  can give rise to interesting developments regarding the possibility  of building self-adjoint operators from a non-real classical function on the phase space while real functions could have non-symmetric quantum operator counterparts.

The organization of the article is as follows. In Section~\ref{mathpb} we present the necessary mathematical material for understanding the ${\mathcal D}$-pb formalism. In Section~\ref{sta-biorthog} we start with a pair of bosonic operators to build orthonormal bases in a Hilbert space. We next make use of f\/inite-dimensional representations of the group ${\rm GL}(2,{\mathbb C})$ to deform this 2-boson algebra into a pair of pseudo-bosons.
In Section~\ref{hpolpb} we illustrate the procedure with deformed complex Hermite polynomials. Note that a set of bi-orthogonal Hermite polynomials were presented, with not much interest in mathematical rigour, in~\cite{tri} via the pseudo-boson operators.  Useful inequalities/estimates  are then proved in Section~\ref{normest}. They are necessary for characterizing ${\mathcal D}$-pb in this ${\rm GL}(2,{\mathbb C})$ context. They are also necessary to get total families in the relevant Hilbert space.
In Section~\ref{bcpb} we introduce two  ``displacement operators'' depending on a  complex variable, and arising as a consequence of the  existence of a pair of ${\mathcal D}$-pb as introduced in Section~\ref{mathpb}. By using these operators we derive two types of ``oblique'' resolutions of the identity (see also~\cite{tri} for previous works in this direction).
Based on these results, we proceed in Section~\ref{quantpb} to the integral  quantization of functions (or distributions)  of a complex variable, obtaining thereby a set of original results. In particular, the quantized version of the canonical Poisson bracket of conjugate pairs $z$ and $\bar z$ is precisely the pseudo-bosonic commutation rule.
In Section~\ref{conclu} we sketch what could be done in the future, starting from the  results presented in this paper. The three appendices  give a set of technical formulae used in the main body of the paper. Those concerning the f\/inite-dimensional  irreducible representations of the group ${\rm GL}(2,{\mathbb C})$ are given in Appendix~\ref{appuir}, while those concerning some of the  asymptotic behaviours of  the corresponding matrix elements are given in Appendix~\ref{asympTs}. Finally, in Appendix~\ref{AmatelMD} we deduce the matrix elements of the bi-displacement operators introduced in connection with bi-coherent states and our pseudo-bosonic integral quantization.

\section[The mathematics of ${\mathcal D}$-pbs]{The mathematics of $\boldsymbol{\mathcal D}$-pbs}
\label{mathpb}

Let ${\mathcal H}$ be a Hilbert space with scalar product $\langle\cdot,\cdot\rangle$ and associated norm $\Vert\cdot\Vert$.   Further, let~$a$ and~$b$ be two operators
on ${\mathcal H}$, with domains $D(a)$ and~$D(b)$ respectively, $a^\dagger$ and $b^\dagger$ their respective adjoints; we assume the existence of a dense set ${\mathcal D}$  in ${\mathcal H}$
such that $a^\sharp{\mathcal D}\subseteq{\mathcal D}$ and $b^\sharp{\mathcal D}\subseteq{\mathcal D}$, where~$x^\sharp$ is either~$x$ or~$x^\dagger$: ${\mathcal D}$ is assumed to be stable under  the
action of~$a$,~$b$, $a^\dagger$ and~$b^\dagger$. Note that we are not requiring here that~${\mathcal D}$ coincide with, e.g., $D(a)$ or~$D(b)$. However
due to the fact that~$a^\sharp f$ is well def\/ined, and belongs to ${\mathcal D}$ for all~$f\in{\mathcal D}$, it is clear that~${\mathcal D}\subseteq D(a^\sharp)$.
Analogously, we conclude that ${\mathcal D}\subseteq D(b^\sharp)$.

\begin{defn}\label{FBdef21}
The operators $(a,b)$ are ${\mathcal D}$-pseudo-bosonic (${\mathcal D}$-pb) if, for all $f\in{\mathcal D}$, we have
\begin{gather}
 abf-baf=f. \label{FB31}
\end{gather}
\end{defn}

To simplify the notation at many places in the sequel,  instead of~(\ref{FB31}) we will simply write $[a,b]=I$, where~$I$ is the identity operator on~${\mathcal H}$
  having in mind that both sides of this equation
have to act on a certain~$f\in{\mathcal D}$.

Our  working assumptions are the following:

\begin{Assumption}\label{Assumption1}
There exists a non-zero $\varphi_{ 0}\in{\mathcal D}$ such that $a\varphi_{ 0}=0$.
\end{Assumption}

\begin{Assumption}\label{Assumption2}
There exists a non-zero $\Psi_{ 0}\in{\mathcal D}$ such that $b^\dagger\Psi_{ 0}=0$.
\end{Assumption}

We then def\/ine the vectors
\begin{gather}
\varphi_n:=\frac{1}{\sqrt{n!}} b^n\varphi_0,\qquad \Psi_n:=\frac{1}{\sqrt{n!}} {a^\dagger}^n\Psi_0, \qquad n\geq0,\label{FB32}
\end{gather} and  introduce the sets ${\mathcal F}_\Psi=\{\Psi_{ n}, \, n\geq0\}$ and ${\mathcal F}_\varphi=\{\varphi_{ n}, \,n\geq0\}$.
Since ${\mathcal D}$ is stable in particular under the action of~$a^\dagger$ and~$b$, we deduce that each $\varphi_n$ and each~$\Psi_n$ belongs to~${\mathcal D}$
and, therefore, to the domains of~$a^\sharp$,~$b^\sharp$ and~$N^\sharp$, where~$N=ba$.

It is now simple to deduce the following lowering and raising relations:
\begin{alignat*}{3}
& b \varphi_n=\sqrt{n+1} \varphi_{n+1}, \qquad && n\geq 0,& \nonumber\\
& a \varphi_0=0,\qquad a\varphi_n=\sqrt{n} \varphi_{n-1}, \qquad &&  n\geq 1,& \nonumber\\
& a^\dagger\Psi_n=\sqrt{n+1} \Psi_{n+1}, \qquad && n\geq 0,& \nonumber\\
& b^\dagger\Psi_0=0,\qquad b^\dagger\Psi_n=\sqrt{n} \Psi_{n-1}, \qquad && n\geq 1,& 
\end{alignat*}
as well as the following eigenvalue equations: $N\varphi_n=n\varphi_n$ and   $N^\dagger\Psi_n=n\Psi_n$, $n\geq0$, where
$N^\dagger=a^\dagger b^\dagger$.  As a consequence
of these  equations,  choosing the normalization of~$\varphi_0$ and~$\Psi_0$ in such a way $\langle \varphi_0,\Psi_0\rangle =1$, we also deduce that
\begin{gather*}
 \langle \varphi_n,\Psi_m\rangle =\delta_{n,m}, 
\end{gather*}
 for all $n, m\geq0$.

The conclusion is, therefore, that ${\mathcal F}_\varphi$ and ${\mathcal F}_\Psi$
are biorthonormal sets of eigenstates of~$N$ and~$N^\dagger$, respectively. This, in principle, does not allow us to conclude that they are also bases  for~${\mathcal H}$, or even Riesz bases. However, let us introduce for the time being the following assumption:

\begin{Assumption}\label{Assumption3}
${\mathcal F}_\varphi$ is a basis for ${\mathcal H}$.
\end{Assumption}

Notice that this automatically implies that  ${\mathcal F}_\Psi$ is a basis for ${\mathcal H}$ as well~\cite{you}. However,  examples are known in which this natural assumption is not satisf\/ied; see, for instance, \cite{bag4,bag2,bag1,bag3, davies,davies2}.
In view of this fact,  a weaker version of Assumption~${\mathcal D}$-pb~\ref{Assumption3} has been introduced recently: for that the concept of
{\em ${\mathcal G}$-quasi bases} is necessary.

\begin{defn}\label{defquasibases}
Let ${\mathcal G}$ be a suitable dense subspace of ${\mathcal H}$. Two biorthogonal sets ${\mathcal F}_\eta=\{\eta_n\in{\mathcal H}$, $n\geq0\}$
and ${\mathcal F}_\Phi=\{\Phi_n\in{\mathcal H},\,n\geq0\}$ are called ${\mathcal G}$-quasi bases if, for all $f, g\in {\mathcal G}$, the following holds:
\begin{gather}
\langle f,g\rangle =\sum_{n\geq0}\langle f,\eta_n\rangle \langle \Phi_n,g\rangle =\sum_{n\geq0}\langle f,\Phi_n\rangle \langle \eta_n,g\rangle. \label{FB35}
\end{gather}
\end{defn}

 Is is clear that, while Assumption ${\mathcal D}$-pb~\ref{Assumption3} implies (\ref{FB35}), the reverse is false. However, if ${\mathcal F}_\eta$ and ${\mathcal F}_\Phi$ satisfy~(\ref{FB35}), we still have some (weak) form of the resolution of the identity  and we can still deduce
interesting results, specially in view of quantization as presented in Section~\ref{quantpb}. For the sake of simplicity, we will often use in the sequel the popular shorthand notation
\begin{gather}
\label{resunipb}
\sum_{n\geq 0} | \eta_n\rangle\langle\Phi_n |= I,
\end{gather}
to be understood in the weak sense on a  dense subspace.
   Incidentally we see that if $f\in {\mathcal G}$ is orthogonal to all the $\Phi_n$'s (or to all the $\eta_n$'s), then $f$
is necessarily zero: we say that ${\mathcal F}_\Phi$ (or ${\mathcal F}_\eta$) is total in~${\mathcal G}$. {Note that this does not imply that these families are total  in the whole Hilbert space ${\mathcal H}$ since we suppose that~\eqref{FB35} holds for $f,g\in{\mathcal G}$, but not, in general, for $f,g\in{\mathcal H}$. Therefore we cannot conclude that each vector of~${\mathcal H}$ orthogonal to, say, all the $\varphi_n$ is necessarily zero, while we can conclude this for each vector of~${\mathcal G}$.

With this in mind, we now consider the aforementioned weaker form of Assumption  ${\mathcal D}$-pb~\ref{Assumption3}:

\begin{AssumptionW}\label{Assumption4}
${\mathcal F}_\varphi$ and ${\mathcal F}_\Psi$ are ${\mathcal G}$-quasi bases, for some dense subspace~${\mathcal G}$  in~${\mathcal H}$.
\end{AssumptionW}

Two important operators, in general unbounded, are the following ones
\begin{gather*}
D(S_\varphi)=\bigg\{f\in{\mathcal H}\colon  \sum_{ n} \langle \varphi_{n},f\rangle \varphi_{n} \ \mbox{exists in} \  {\mathcal H} \bigg\}\qquad \mbox{and}\qquad S_\varphi f=\sum_{ n} \langle \varphi_{n},f\rangle ,\varphi_{n}
\end{gather*}
for all $f\in D(S_\varphi)$, and, similarly,
\begin{gather*}
D(S_\Psi)=\bigg\{h\in{\mathcal H}\colon  \sum_{ n} \langle \Psi_{n},h\rangle \Psi_{ n} \ \mbox{exists in}  \ {\mathcal H} \bigg\}\qquad \mbox{and}\qquad S_\Psi h=\sum_{ n} \langle \Psi_{n},h\rangle\Psi_{ n},
\end{gather*}
 for all $h\in D(S_\Psi)$. It is clear that $\Psi_n\in D(S_\varphi)$ and $\varphi_n\in D(S_\Psi)$, for all $n\geq0$.
 However, since ${\mathcal F}_\varphi$ and ${\mathcal F}_\Psi$ are not required to be bases here, it is convenient to work under the additional
 hypothesis that ${\mathcal D}\subseteq D(S_\Psi)\cap D(S_\varphi)$, which is often true in concrete examples~\cite{bagbook}. In this way~$S_\Psi$ and~$S_\varphi$ are automatically densely def\/ined.
 Also, since $\langle S_\Psi f,g\rangle =\langle  f,S_\Psi g\rangle $ for all $f,g\in D(S_\Psi)$, $S_\Psi$ is a symmetric operator,
 as well as $S_\varphi\colon \langle S_\varphi f,g\rangle =\langle f,S_\varphi g\rangle $ for all $f,g\in D(S_\varphi)$.
 Moreover, since they are positive operators, they are also semibounded~\cite{ped}
\begin{gather*}
 \langle S_\varphi f,f\rangle \geq 0,\qquad \langle S_\Psi h,h\rangle \geq0,
 \end{gather*}
for all $f\in D(S_\varphi)$ and $h\in D(S_\Psi)$. Hence both these operators admit  self-adjoint (Friedrichs) extensions,  $\hat S_\varphi$ and
$\hat S_\Psi$~\cite{ped}, which are both also positive. Now, the spectral theorem ensures us that we can def\/ine the
square roots $\hat S_\Psi^{1/2}$ and $\hat S_\varphi^{1/2}$, which are
self-adjoint and positive and, in general, unbounded. These operators can be used to def\/ine new scalar products and new related notions of the  adjoint, as well as new mutually  orthogonal vectors. These and other aspects, which are particularly relevant in the present context, are discussed in some detail in~\cite{bagbook}.

\section{Biorthogonal families of vectors and polynomials}\label{sta-biorthog}

In this section we present the f\/irst illustration of the above formalism with an explicit group theoretical construction of pseudo-bosonic operators.
We start with a pair of bosonic operators~$a_i$, $a_i^\dag$, $i = 1,2$, acting ({\em irreducibly}) on the Hilbert space~$\mathcal H$. They satisfy the commutation relations,
\begin{gather}
\big[a_i, a_j^\dag \big] = I\delta_{ij}, \qquad    [a_i, a_j ] = \big[a_i^\dag , a_j^\dag \big] = 0, \qquad i,j = 1,2 .
\label{2dim-ccr}
\end{gather}
Starting with the (normalized) ground state vector $\varphi_{0,0}$, for which $a_i \varphi_{0,0} = 0$, $i=1,2$, we def\/ine the vectors,
\begin{gather}
 \varphi_{n_1,n_2} = \frac {\big(a_1^\dag\big)^{n_1}  \big(a_2^\dag\big)^{n_2}}{\sqrt{n_1!  n_2!}}\varphi_{0,0}, \qquad n_1, n_2 =0,1,2, \ldots , \infty .
 \label{stan-onb}
\end{gather}
 These vectors form an orthonormal basis in~$\mathcal H$.

 We now reorder the elements  of this basis as in~\eqref{L-subsp-vect} below. For any integer~$L\geq 0$, let us def\/ine the set of~$L+1$ vectors
   \begin{gather}
     f^L_m =  \frac {\big(a_1^\dag\big)^m  \big(a_2^\dag\big)^{L-m}}{\sqrt{m!  (L-m)!}}\varphi_{0,0} = \varphi_{m, L-m}, \qquad m =0,1,2, \ldots, L,
 \label{L-subsp-vect}
 \end{gather}
and denote by $\mathcal H^L$ the $(L+1)$-dimensional subspace of $\mathcal H$ spanned by these vectors. Clearly
\begin{gather*}
  \langle f^L_m , f^M_n\rangle = \delta_{LM} \delta_{mn} \qquad \text{and} \qquad
\mathcal H = \oplus_{L=0}^\infty \mathcal H^L .
\end{gather*}
Hence,  the $f^L_m $ are  a relabeling of the vectors~$\varphi_{n_1,n_2}$ which will be useful in the sequel.

\subsection{A second basis and a Cuntz algebra}\label{sec-bas-cuntz}

Using the vectors $f^L_m$, we now introduce a second relabeling, this time using a single index. We set
\begin{gather}
  F_n = f^L_m  =  \varphi_{m, L-m}, \qquad \text{where} \qquad  n = \frac {L(L+1)}2 + m .
\label{F-vects}
\end{gather}
Note that in making this relabeling, we have used the bijective map $\beta\colon \mathbb N\times  \mathbb N \to  \mathbb N$, def\/ined by
\begin{gather}
 n = \beta (n_1 , n_2 ) = \frac {(n_1 + n_2)(n_1 +n_2 +1)}2 + n_1.
\label{bijec}
\end{gather}
The inverse map $(n_1, n_2 ) = \beta^{-1} (n)$ is obtained by taking
\begin{gather*}
 L = \sup_{\ell\in \mathbb N}  \left\{\ell\colon \frac {\ell(\ell+1)}2 \leq n \right\}
 \end{gather*}
and then writing
\begin{gather*}
n_1 = n- \frac {L(L+1)}2 \qquad \text{and} \qquad  n_2 = L - n_1.
\end{gather*}
These successive relabelings of the point set ${\mathbb N}^2$ are }graphically described in Fig.~\ref{relabel} below.  They just illustrate the well-known countability of ${\mathbb N}^2$, or, equivalently, of the positive rational numbers.
\begin{figure}[htb!]
\centering
\includegraphics{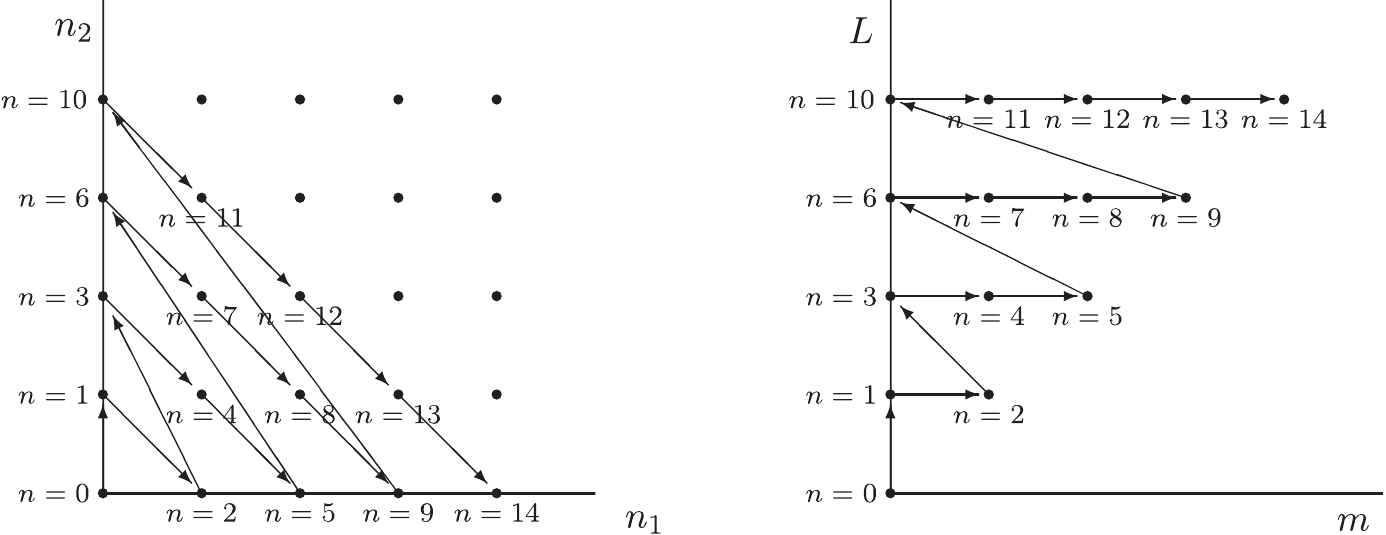}

 \caption{Three successive relabelings of the point set ${\mathbb N}^2$. On the left:  point set ${\mathbb N}^2$ of pairs $(n_1,n_2)$ with their corresponding non-negative label $n$. On the right $2d$-relabeling $(n_1,n_2) \mapsto (m,L)$ where $m=n_1$, $L=n_1+n_2$ and eventually  $1d$-relabeling $n= \frac{L(L+1)}{2} + m$ where dots in the sector $\{(m,L) , \, L\in {\mathbb N}  , \, m=0, 1, \dots , L\}$ are marked with their corresponding  non-negative label $n$.}
\label{relabel}
\end{figure}

We next def\/ine two bosonic operators $B$, $B^\dag$, in the standard manner, using the vectors $F_n$:
\begin{gather}
   BF_n = \sqrt{n}F_{n-1}, \qquad BF_0 = 0, \qquad B^\dag F_n = \sqrt{n+1} F_{n+1} ,
\qquad [B,B^\dag] = I,
\label{F-bos-op}
\end{gather}
and from (\ref{F-vects}) we f\/ind their actions on the vectors $f^L_m$:
\begin{gather}
 Bf^L_m   =  \begin{cases} \sqrt{\dfrac {L(L+1)}2 + m} f^L_{m-1}, & \text{if} \  m> 0,\vspace{1mm}\\
                           \sqrt{\dfrac {L(L+1)}2 } f^{L-1}_{L-1}, &\text{if}\  m =0 ,\end{cases} \nonumber\\
  B^\dag f^L_m  =  \begin{cases} \sqrt{\dfrac {L(L+1)}2 + m+1} f^L_{m+1}, & \text{if} \  m<L,\vspace{1mm}\\
                           \sqrt{\dfrac {(L+1)(L+2)}2 } f^{L+1}_0, &\text{if}\  m =L. \end{cases}
 \label{f-bos-op}
\end{gather}
This means that, writing the vectors $f^L_m$ and $F_n$ in ascending order
\begin{gather}
\label{mapfF}
\begin{array}{@{}c} f^0_0\\ \updownarrow\\F_0\end{array}, \underbrace{\begin{array}{c}f^1_0\\ \updownarrow \\ F_1\end{array}, \begin{array}{c}f^1_1\\ \updownarrow \\F_2\end{array}}, \underbrace{\begin{array}{c}f^2_0\\ \updownarrow \\ F_3\end{array}, \begin{array}{c}f^2_1\\ \updownarrow\\F_4\end{array}, \begin{array}{c}f^2_2\\ \updownarrow\\F_5\end{array}}, \underbrace{\begin{array}{c}f^3_0\\ \updownarrow\\F_6\end{array}, \begin{array}{c}f^3_1 \\ \updownarrow\\F_7\end{array}, \begin{array}{c}f^3_2\\ \updownarrow\\F_8\end{array}, \begin{array}{c}f^3_3\\ \updownarrow\\F_9\end{array}},\ldots  ,
\end{gather}
the operators $B^\dag$ and $B$  move them up and down this array, respectively (see right part of Fig.~\ref{relabel}).

As a direct consequence of the maps $(m,L) \mapsto n$ introduced in~\eqref{F-vects} and the above correspondence~\eqref{mapfF}, there is an interesting set of isometries $S_n$, $n = 0, 1, 2, \ldots , \infty$, of the Hilbert space $\mathcal H$ associated to the two sets of basis vectors~$\{F_n\}$ and $\{\varphi_{n,m}\}$. We def\/ine these operators as
\begin{gather}
S_n F_m = \varphi_{m,n} = F_{k(m,n)},
\label{isometries}
\end{gather}
where
\begin{gather*}
k(m,n) := \frac {(m+n)(m+n+1)}2 +m, \qquad n,m = 0,1,2, \ldots , \infty.
\end{gather*}
Clearly, $\Vert S_n\Vert = 1$, $n = 0,1,2, \ldots , \infty$. These operators were introduced in~\cite{abr}, where they were used to construct coherent states on $C^*$-Hilbert modules. The following properties are easily proved.

\begin{prop}\quad
\begin{enumerate}\itemsep=0pt
  \item[$(i)$]
The isometries $S_n$  are not unitary maps. Indeed, one has
\begin{gather*}
S_m^\dagger S_n =
\delta_{mn} I  \qquad \text{and} \qquad
S_n S_n^\dagger   =  \mathbb P_n ,
\end{gather*}
$\mathbb P_n$ being the projection
operator onto the subspace $\mathcal H_n$ of $\mathcal H$ spanned by the vectors
$\varphi_{m, n}$, $m = 0,1,2, \ldots , \infty$.
\item[$(ii)$] The kernel of $S_n^\dagger$ is the set of all vectors of the type $\varphi_{m,k}$, $m = 0, 1, 2, \ldots$,  and $k \neq n$.
\item[$(iii)$] $S_mS_n^\dagger$ is a partial isometry from~$\mathcal H_n$ to~$\mathcal H_m$.
\item[$(iv)$] The positive  operators $S_n S_n^\dagger$ resolve the identity
\begin{gather}
\sum_{n=0}^\infty S_n S_n^\dagger = I  ,
\label{isom2}
\end{gather}
the sum converging strongly.
  \item[$(v)$] There exist the following  relationships between  the operators $a_1$, $a_1^\dag$ in~\eqref{2dim-ccr} and the ope\-ra\-tors~$B$, $B^\dag$ in~\eqref{f-bos-op} through $S_n$, $S_n^\dagger$:
\begin{gather*}
  S_n^\dagger a_1 S_n = B, \qquad   S_n^\dagger a_1^\dag S_n = B^\dag,
\end{gather*}
while $S_n^\dagger a_2 S_n = S_n^\dagger a_2^\dag S_n = 0$.
  \end{enumerate}
  \end{prop}

The $S_n$ generate a $C^*$-algebra $\mathcal O_\infty$, known as a Cuntz algebra~\cite{Cuntz}, which is a subject of independent interest. Note also, that we have used here a very specif\/ic bijection~(\ref{bijec}) to def\/ine the vectors~$F_n$. Of course, there are many other possible bijections, which will also give rise to associated Cuntz algebras. But this particular one will be useful for our subsequent analysis.

\subsection{Deformed operators and bases}\label{def-op-basis}

To proceed further, let
\begin{gather*}
   g = \begin{pmatrix} g_{11} & g_{12} \\ g_{21} & g_{22} \end{pmatrix},
\end{gather*}
be an element of the GL$(2, \mathbb C)$ group (i.e., $g$ is a complex $2\times 2$ matrix with $\det [g] \neq 0$), using which we def\/ine the new operators
\begin{gather*}
 A^g_1  = \bar g_{11}a_1 + \bar g_{21}a_2 , \qquad  A^g_2  = \bar g_{12}a_1 + \bar g_{22}a_2 ,
\end{gather*}
and the corresponding adjoint operators  $A^{g \dag}_i$, $i=1,2$,
i.e., in matrix notations
\begin{gather*}
\begin{pmatrix}
  A^g_1       \\
    A^g_2
\end{pmatrix} \equiv \mathbf{A}^g= g^{\dag}\cdot \mathbf{a}  , \qquad \mathbf{a}:= \begin{pmatrix}
  a_1       \\
    a_2
\end{pmatrix}  , \qquad \begin{pmatrix}
  A^{g \dag}_1       \\
    A^{g \dag}_2
\end{pmatrix} \equiv \mathbf{A}^{g +}= {}^tg\cdot \mathbf{a}^+  , \qquad \mathbf{a}^+:= \begin{pmatrix}
  a^{\dag}_1       \\
    a^{\dag}_2
\end{pmatrix}   .
\end{gather*}
We call these operators {\em deformed bosonic operators}; they satisfy  $[A^g_1, A^g_2 ] = [A^{g\dag}_1, A^{g\dag}_2 ] = 0$, however, the other commutators are in general dif\/ferent from those of the undeformed operators $a_i$, $a_j$, $i=1,2$. Indeed, we have the general commutation relations
\begin{gather*}
  [A_i^g , A^{g \dag}_j] = \overline{g_{1i}} g_{1j} + \overline{g_{2i}} g_{2j}  , \qquad  i,j =1,2,
\end{gather*}
so that the matrix elements of  $g$ would have  to satisfy $\overline{g_{1i}}\; g_{1j} + \overline{g_{2i}} g_{2j} = \delta_{ij}$ (which is equivalent to having $g^\dag g = I_2$, i.e., a~unitary matrix) in order to recover the standard commutation rela\-tions~(\ref{2dim-ccr}). However we leave aside this condition, which is not relevant for us.

Using the operators $A^{g\dag}_i$, $i=1,2$, and noting that $A^g_i\varphi_{0,0} = 0$, we now construct a set of $g$-{\em deformed basis vectors} in a manner analogous to the construction of the vectors~$\varphi_{n_1, n_2}$ in~(\ref{stan-onb}). We def\/ine
\begin{gather}
 \varphi_{n_1,n_2}^g = \frac {(A^{g \dag}_1)^{n_1}  (A^{g \dag}_2)^{n_2}}{\sqrt{n_1!  n_2!}}\varphi_{0,0}, \qquad n_1, n_2 =0,1,2, \ldots , \infty   .
 \label{g-def-basis}
\end{gather}
Adopting the group representation notations 
 \eqref{monbas}, we rewrite \eqref{g-def-basis} as
  \begin{gather}
\label{g-def-basis1}
\varphi_{n_1,n_2}^g = e_{n_1,n_2}\big({}^t g\cdot\mathbf{a}^+\big) \varphi_{0,0}, \qquad  e_{n_1,n_2}(\mathbf{a}^+) := \frac{{a^{\dag}_1}^{n_1}{a^{\dag}_2}^{n_2}}{\sqrt{n_1! n_2!}}   .
\end{gather}
It is obvious that, in general, these vectors are not mutually orthogonal, since they are not eigenstates (with dif\/ferent eigenvalues) of some self-adjoint operator. To continue, for each~$L\geq 0$ let us def\/ine the set of~$L+1$ vectors~$f^{g, L}_m$ in a manner analogous to~(\ref{L-subsp-vect})
\begin{gather}
     f^{g,L}_m = e_{m,L-m}({}^tg\cdot\mathbf{a}^+)\varphi_{0,0} = \varphi_{m, L-m}^g, \qquad m =0,1,2, \ldots, L,
 \label{def-L-subsp-vect}
 \end{gather}
It is clear that these vectors are linear combinations of the $f^L_m$, hence they also span the subspace~$\mathcal H^L$ of $\mathcal H$. This is simply due to the ${\rm GL}(2,{\mathbb C})$ representation operator def\/ined in~\eqref{repglpol} with $L=s$, and acting in the present context as the map~\cite{bsali}  $\mathcal T^L (g)\colon \mathcal H^L \to \mathcal H^L$ for which
\begin{gather}
   \mathcal T^L (g) f^L_m = f^{g,L}_m, \qquad m = 0,1,2, \ldots, L, \qquad g \in \mathrm{GL}(2,\mathbb C).
\label{rep-op}
\end{gather}
 The matrix elements of the operators $\mathcal T^L (g)$ in the basis~\eqref{g-def-basis1} are given in~\eqref{matelgamapp}.
 In the~$f^L_m$ basis they read as~\cite{bsali}
\begin{gather}
\mathcal T_{m^{\prime}m}^L (g) = \sum_{q}\binom{m}{q}\binom{L-m}{m^{\prime}-q} g_{11}^{q}g_{21}^{m-q}g_{12}^{m^{\prime}-q}g_{22}^{L-m+q-m^{\prime}}, \qquad  0\leq m^{\prime},m \leq L  .
\label{irrep-mat-elem}
\end{gather}

The range of values assumed by $q$ in the above sum is determined by the cancellation of the  binomial coef\/f\/icients involved, i.e.,  $\max\{0,m^{\prime}+m-L\}\leq q \leq\min\{m^{\prime},m\}$.
These matrix elements  are discussed in greater detail in  Section~\ref{normest}.

\subsection{Biorthogonal families of vectors and pseudobosons}\label{subsec-biorthog}

 Corresponding to the vectors $f^{g,L}_m$, let us def\/ine a dual family of vectors $\widetilde{f}^{g,L}_m$ by the relation
\begin{gather}
  \widetilde{f}^{g,L}_m  = \mathcal T^L(\widetilde{g})f^L_m = f^{\widetilde{g},L}_m , \qquad \widetilde{g} := (g^\dag)^{-1} .
\label{dual-vects}
\end{gather}
Clearly these vectors are also elements of the subspace $\mathcal H^L$. From (\ref{rep-op}), (\ref{star-rep}) and the representation theoretical property of $\mathcal T^L (g)$, by which  $\mathcal T^L (g^{-1}) = (\mathcal T^L (g))^{-1}$, we see that
\begin{gather*}
  \big\langle  \widetilde{f}^{g,L}_m , f^{g,M}_n \big\rangle = \delta_{LM} \delta_{mn} .
\end{gather*}
This means that on each subspace $\mathcal H^L$ the vectors $f^{g,L}_m$ and $\widetilde{f}^{g,L}_m$ form two {\em biorthogonal bases}, while they are, in general, biorthogonal sets in~${\mathcal H}$.

Consider now the operator $\mathcal T (g) = \oplus_{L=0}^\infty \mathcal T^L (g)$. This operator is in general unbounded and densely def\/ined in~$\mathcal H$, since~$\mathcal T^L (g)$ is bounded on each subspace~$\mathcal H^L$. In particular $\mathcal T(g)$ is well def\/ined on the vectors $F_n$ in~(\ref{F-vects}). We thus def\/ine the two sets of vectors
\begin{gather*}
F^g_n = \mathcal T (g) F_n, \qquad \text{and} \qquad \widetilde{F}^g_n  = \mathcal T(\widetilde{g})F_n = F^{\widetilde g}_n, \qquad  n=0,1, \ldots , \infty,
\end{gather*}
in duality, for which
\begin{gather*}
   \big\langle \widetilde{F}^g_m, F^g_n \big\rangle = \delta_{mn} .
\end{gather*}
Note that the existence of the inverse operator $(\mathcal T(g))^{-1}$, as a densely def\/ined operator on~$\mathcal H$ is guaranteed by the  property $(\mathcal T^L(g))^{-1} = \mathcal T^L(g^{-1})$ on each subspace~$\mathcal H^L$.

It is now possible to construct families of pseudobosons using the vectors~$F^g_n $ and~$\widetilde{F}^g_n$.
The following
proposition is easily derived from the above material.

\begin{prop}
Given the operators $B$, $B^\dag$ in \eqref{F-bos-op}, for any $g \in {\rm GL}(2, {\mathbb C})$ let us define  the deformed operators
\begin{gather*}
  B(g) = \mathcal T(g) B (\mathcal T (g))^{-1}, \qquad \widetilde B (g) = B(\widetilde{g}) ,
\end{gather*}
and their adjoints $B(g)^\dag$, $\widetilde B (g)^\dag$. Then, as operators on the full Hilbert space $\mathcal H$, they satisfy, at least formally, the pseudo-bosonic commutation relations
\begin{gather*}
  \big[B(g) , \widetilde B (g)^\dag\big] = \big[\widetilde B (g), B(g)^\dag\big] = I.
\end{gather*}
Their actions on the vectors $F^g_n$, $\widetilde{F}^g_n$ read as
\begin{alignat*}{3}
& B(g)F^g_n = \sqrt{n} F^g_{n-1} , \qquad && B(g)^\dag \widetilde F^g_n   =   \sqrt{n+1}\widetilde F^g_{n+1} ,& \\
& \widetilde B(g)\widetilde F^g_n = \sqrt{n} \widetilde F^g_{n-1} , \qquad && \widetilde B(g)^\dag F^g_n   =   \sqrt{n+1} F^g_{n+1} .&
\end{alignat*}
\end{prop}

Notice that, all throughout this section, $g$ is a f\/ixed element in ${\rm GL}(2, {\mathbb C})$. This is important since, if we take
$g_1, g_2\in {\rm GL}(2, {\mathbb C})$, with $g_1\neq g_2$, then nothing can be said about $[B(g_1) , \widetilde B (g_2)^\dag]$, for instance.

To relate the equations above with the general structure discussed in Section~\ref{mathpb}, we start by observing that $B(g)F^g_0 =0= \widetilde B(g)\widetilde F^g_0$. This shows that the two non zero vacua required in Assumptions~${\mathcal D}$-pb~\ref{Assumption1} and ${\mathcal D}$-pb~\ref{Assumption2} of Section~\ref{mathpb} do exist and coincide\footnote{Here $B(g)$ and $\widetilde B(g)$ respectively play the role of $a$ and $b$ in~Section \ref{mathpb}.}. In fact $F^g_0 = \widetilde F^g_0 =\varphi_{0,0}$. Moreover, calling ${\mathcal D}$ the linear span of the vectors $\varphi_{n_1,n_2}$ in~(\ref{stan-onb}), it is clear that (i)~$F^g_0$, $\widetilde F^g_0\in {\mathcal D}$, (ii)~that~${\mathcal D}$ is dense in ${\mathcal H}$ and (iii)~${\mathcal D}$ is left invariant by $B(g)$, $\widetilde B(g)$ and by their adjoints. In fact these operators map each f\/inite linear combination of the $\varphi_{n_1,n_2}$'s into a dif\/ferent, but still f\/inite, linear combination of the same vectors. Thus we conclude that the present setup ref\/lects, at least in part, what was discussed in Section \ref{mathpb}. However, in Section~\ref{normest} we will show that Assumption~${\mathcal D}$-pb~\ref{Assumption3} is not satisf\/ied, while its weaker version Assumption~${\mathcal D}$-pbw~\ref{Assumption4} holds true.

\section{Deformed complex Hermite polynomials}\label{hpolpb}

In this section we give a concrete realization of the kind of pseudo-bosons discussed above. Let us consider the
irreducible representation of the operators $a_i$, $ a_i^\dag$, $i =1,2$, on the Hilbert space $\mathcal H({\mathbb C})={\mathcal L}^2({\mathbb C}, \mathrm{d}\nu(z,\overline z))$, where
\begin{gather*}
\mathrm{d}\nu(z,\overline z)=e^{-|z|^2}\frac{\mathrm{d} z\wedge \mathrm{d}\overline{z}}{2\pi i}=\frac{1}{\pi}e^{-(x^2+y^2)}\mathrm{d} x \mathrm{d} y , \qquad z=x+iy ,
\end{gather*}
where they are realized as follows
\begin{gather}
a_1=\partial_z,\quad a_1^\dagger=z-\partial_{\overline z},\qquad a_2=\partial_{\overline z},\quad a_2^\dagger={\overline z}-\partial_{ z}  .
\label{compl=rep}
\end{gather}
The basis vectors $\varphi_{n_1, n_2}$, given in~(\ref{stan-onb}), now turn out to be the normalized {\em complex Hermite polynomials} in the variables $z$, $\overline{z}$, which we shall denote by $h_{n_1, n_2}({\mathfrak{z}})$, where we adopt the vector notation for group theoretical reasons
\begin{gather*}
{\mathfrak{z}}:= \begin{pmatrix}
      z   \\
      \bar z
\end{pmatrix}.
\end{gather*}
The normalized vacuum state $\varphi_{0,0}$, satisfying $a_i \varphi_{0,0} = 0$, $i=1,2,$ is simply the constant function $h_{0,0}({\mathfrak{z}})=1$. These polynomials have been discussed extensively in the literature (see, for example, \cite{albaga12,cotgazgor10,Gha,Ito}, and very recently in~\cite{ghanmi15, Ism1,ismail15,Ism2,Ism3,Ism4}). Their expression can be directly inferred from~\eqref{stan-onb}
\begin{gather}
\nonumber
   h_{n_1,n_2}({\mathfrak{z}}) =\frac {(z-\partial_{\overline z})^{n_1}  ({\overline z}-\partial_{ z})^{n_2}}{\sqrt{n_1!  n_2!}} h_{0,0}\\
\hphantom{h_{n_1,n_2}({\mathfrak{z}}) =}{}
   =\frac 1{\sqrt{n_1! n_2!}}\sum_{k=0}^{\min(n_1,n_2)}(-1)^k k!{n_1\choose k}{n_2\choose k}z^{n_1-k} \overline{z}^{n_2-k}  .
\label{comp-herm-poly}
\end{gather}

 Alternatively, they can also be obtained from the expression
\begin{gather}
h_{n_1,n_2}({\mathfrak{z}})  =  e^{-\partial_z \partial_{ \overline{z}}}  \frac{z^{n_1} \overline{z}^{n_2}}{\sqrt{n_1! n_2!}}= e^{-\partial_z \partial_{ \overline{z}}} e_{n_1,n_2}(\mathbf{\mathfrak{z}})  .
\label{comp-herm-poly2}
\end{gather}
Note that these complex Hermite polynomials are of particular interest in the study of physical systems constituted by several layers. Indeed, such systems can be modeled by spaces of polyanalytic functions generated by complex Hermite polynomials. This has recently found several applications in signal analysis \cite{abreu10,abreu12,grochenig09} and in the statistics of higher Landau levels~\cite{haimi13}.

The $g$-deformed basis vectors $\varphi^g_{n_1,n_2}$ in (\ref{g-def-basis}), which we now denote by $h^g_{n_1,n_2}$, are also polynomials in~$z$, $\overline{z}$, which are linear combinations of the $h_{n_1,n_2}$~\cite{alimourshah,bsali,Wun}. Within the ${\rm GL}(2,{\mathbb C})$ representation framework, they are obtainable from a~formula analogous to~(\ref{g-def-basis1})
\begin{gather}
\label{defhLnn}
h^{g,L}_{n_1,n_2}({\mathfrak{z}})  = e^{-\partial_z \partial_{ \overline{z}}} e_{n_1,n_2}\big({}^tg\cdot{\mathfrak{z}}\big) := e^{-\partial_z \partial_{ \overline{z}}}\big(\mathcal{T}^L(g) e_{n_1,n_2}\big)({\mathfrak{z}})  , \qquad L= n_1 + n_2 .
\end{gather}
Similarly, with the notation introduced in~\eqref{dual-vects}, we def\/ine the dual polynomials
\begin{gather}
\label{defThLnn}
\widetilde{h}^{g,L}_{n_1,n_2}({\mathfrak{z}})  = h^{\widetilde{g},L}_{n_1,n_2}({\mathfrak{z}})  =e^{-\partial_z \partial_{ \overline{z}}} e_{n_1,n_2}\big({}^t\widetilde{g}\cdot{\mathfrak{z}}\big) := e^{-\partial_z \partial_{ \overline{z}}} \big(\mathcal{T}^L(\widetilde{g}\big) e_{n_1,n_2})({\mathfrak{z}})  .
\end{gather}
We derive from  \eqref{comp-herm-poly2} and the def\/inition given in~\eqref{defmatel} of the matrix elements of the representation operator~$\mathcal T^L (g)$, the following expansions in which the apparent double summation is actually reduced a single summation because of the restriction $n_1+n_2 = L= n^{\prime}_1 + n^{\prime}_2$
\begin{gather*}
\begin{split}
& h^{g,L}_{n_1,n_2}({\mathfrak{z}})   = \sum_{n'_1,n'_2=L-n^{\prime}_1}\mathcal{T}^L_{n'_1,n'_2;n_1,n_2}(g)  h_{n^{\prime}_1,n^{\prime}_2}({\mathfrak{z}})  ,\\
& \widetilde{h}^{g,L}_{n_1,n_2}({\mathfrak{z}})  = \sum_{n'_1,n'_2=L-n^{\prime}_1}\mathcal{T}^L_{n'_1,n'_2;n_1,n_2}(\widetilde{g})  h_{n^{\prime}_1,n^{\prime}_2}({\mathfrak{z}})  .
\end{split}
\end{gather*}
Similarly, writing now $h^L_m ({\mathfrak{z}})$ for the relabeled vectors~$f^L_m$ in~(\ref{L-subsp-vect}) and~$h^{g,L}_m$ for the $f^{g,L}_m$ in~(\ref{def-L-subsp-vect}) and using~(\ref{rep-op}) and~(\ref{irrep-mat-elem}) we get
\begin{gather*}
  h^{g,L}_m ({\mathfrak{z}}) = \sum_{m^{\prime}=0}^L \mathcal T^L (g)_{m^{\prime}m} h^L_{m^{\prime}} ({\mathfrak{z}})   , \qquad
  h^{\widetilde{g},L}_m ({\mathfrak{z}}) = \sum_{m^{\prime}=0}^L \mathcal T^L (\widetilde{g})_{m^{\prime}m}\;h^L_{m^{\prime}} ({\mathfrak{z}}).
\end{gather*}

 We refer to the polynomials $h^{g,L}_m ({\mathfrak{z}})$ as {\em deformed complex Hermite polynomials.} It is now a routine matter to go over to a basis $H_n$, $n=0,1,2, \ldots , \infty,$ which would be the analogous relabeling of the $h^L_m$ as the $F_n$ in~(\ref{F-vects}) are the relabeled versions of the $f^L_m$. Similarly we may def\/ine the deformed polynomials $H^g_n ({\mathfrak{z}})$ and $\widetilde{H}^g_n ({\mathfrak{z}}) = H^{\widetilde{g}}_n ({\mathfrak{z}})$. The biorthonormality of these polynomials is expressed via the integral relation
\begin{gather*}
  \int_{\mathbb C} \overline{\widetilde{H}^g_n ({\mathfrak{z}})} H^g_{n^{\prime}} ({\mathfrak{z}})  d\nu(z,\overline z) = \delta_{nn^{\prime}},
\end{gather*}
which then are the pseudo-bosonic complex polynomial states.

Before leaving this section let us note that for f\/ixed $n$, the polynomials $h_{m,n}({\mathfrak{z}})$, $m = 0,1,2, \ldots , \infty$ (see~(\ref{comp-herm-poly})) are {\em polyanalytic functions} of order $n$ (see, for example~\cite{abfeich}).
More precisely, the subspace~$\mathcal H_{n}$ of $\mathcal H (\mathbb C )$, consisting of such polyanalytic functions spanned by the complex Hermite polynomials with a f\/ixed degree is the $n$th polyanalytic sector (corresponding to the $n$th Landau level, also called a~true~\cite{vasilevski97-00} or pure~\cite{haimi13} polyanalytic space) and the direct sum of the f\/irst $k$ such spaces is the $k$th polyanalytic space. A detailed proof of this fact, also valid for Banach spaces, can be found in~\cite{abreu12}. It should also be mentioned that the decomposition of $\mathcal H (\mathbb C )$ into true-polyanalytic spaces was f\/irst introduced by Vasilevski in~\cite{vasilevski97-00}, where the action of the operators in~\eqref{compl=rep} was explored.

Since by~(\ref{isometries}) $S_n H_m = h_{m,n}$, the isometry~$S_n$ maps the whole Hilbert space $\mathcal H (\mathbb C)$ to  its $n$th polyanalytic sector and~(\ref{isom2}) is then the statement that $\mathcal H (\mathbb C)$ decomposes into an orthogonal direct sum of polyanalytic subspaces. (Note that a function $f(z, \overline{z})$ is polyanalytic of order~$n$ if $\frac{\partial^{n+1} f}{\partial \overline{z}^{n+1}} = 0$). Such functions have also found much use recently in signal analysis.

\section{Norm estimates for biorthogonal families of polynomials}
\label{normest}

In this section, we take advantage of the group representation properties and of the orthonor\-ma\-li\-ty of the complex Hermite polynomials to estimate the respective norms of these new (deformed) vectors. This is useful in determining whether the sets $\{h^{g,L}_{n_1,n_2}\}$ and $\{h^{\widetilde g,L}_{n_1,n_2}\}$ constitute bases in the Hilbert space $\mathcal H({\mathbb C})={\mathcal L}^2({{\mathbb C}}, \mathrm{d}\nu(z,\overline z))$ or not. (Some similar, though less sharp, estimates were given in~\cite{bsali}.)

\begin{prop}
Let $g\in \mathrm{GL}(2,{\mathbb C})$ and $\big\{h^{g,L}_{n_1,n_2}\big\}$ and $\big\{h^{\widetilde g,L}_{n_1,n_2}\big\}$ be the deformed complex Hermite polynomials defined in \eqref{defhLnn} and~\eqref{defThLnn} respectively.
\begin{enumerate}\itemsep=0pt
  \item[$(i)$] We have the following upper bounds for their respective norms
\begin{gather}
\label{upboundphipsi}
\big\Vert h^{g,L}_{n_1,n_2}\big\Vert^2 \leq
   \big(\operatorname{tr} g^{\dag} g\big)^{(n_1+n_2)/2}  , \qquad \big\Vert \widetilde{h}^{g,L}_{n_1,n_2} \big\Vert^2 \leq
   \big(\operatorname{tr} \big(g^{\dag} g\big)^{-1}\big)^{(n_1+n_2)/2}  .
\end{gather}
\item[$(ii)$] More precisely
\begin{gather}
\label{lowupboundMT1}
  \frac{a^{n_1}d^{L-n_1}}{\sqrt{\pi \min(n_1,L-n_1)}}\leq \big\Vert h^{g,L}_{n_1,n_2}\big\Vert^2 \leq \binom{L}{n_1} a^{n_1}  d^{L-n_1}   ,\\
\label{lowupboundMT2}   \frac{1}{\vert \det g\vert^{2L}}\frac{d^{n_1}a^{L-n_1}}{\sqrt{\pi \min(n_1,L-n_1)}}\leq \big\Vert \widetilde{h}^{g,L}_{n_1,n_2}\big\Vert^2 \leq\frac{1}{\vert \det g\vert^{2L}} \binom{L}{n_1}   d^{n_1}a^{L-n_1}    ,
\end{gather}
where we have introduced the notations
\begin{gather*}
g^{\dag}g \equiv \begin{pmatrix}
 a     &   b \\
 \bar b     & d
\end{pmatrix}\quad  \Rightarrow \quad \big(g^{\dag}g\big)^{-1}= \frac{1}{\vert \det g\vert^2}\begin{pmatrix}
 d    &  - b \\
 - \bar b     & a
\end{pmatrix}  ,  \qquad \vert \det g\vert^2 = ad- \vert b\vert^2  ,
\end{gather*}
for the positive invertible matrix $g^{\dag}g$.
\end{enumerate}
\end{prop}

\begin{proof}
By using group representation properties for the operator $ \mathcal{T}^L$,  we f\/ind the following expressions for the squared norms of the deformed complex Hermite polynomials
\begin{gather}
\big\Vert h^{g,L}_{n_1,n_2}\big\Vert^2  = \sum_{n'_1,n'_2}\big\vert \mathcal{T}^L_{n'_1,n'_2;n_1,n_2}(g)\big\vert^2
 = \sum_{n'_1,n'_2}\mathcal{T}^s_{n'_1,n'_2;n_1,n_2}(g) \overline{\mathcal{T}^L_{n'_1,n'_2;n_1,n_2}(g)} \nonumber\\
 \hphantom{\big\Vert h^{g,L}_{n_1,n_2}\big\Vert^2}{}
  = \sum_{n'_1,n'_2}\mathcal{T}^L_{n_1,n_2;n'_1,n'_2}\big(\bar{g}^{\dag}\big) \mathcal{T}^L_{n'_1,n'_2;n_1,n_2}(\bar g)
= \mathcal{T}^L_{n_1,n_2;n_1,n_2}\big(\bar{g}^{\dag}\bar g\big)= \mathcal{T}^L_{n_1,n_2;n_1,n_2}\big(g^{\dag}g\big),\!\!\!\label{normphi}
\\
\label{normpsi}
\big\Vert \widetilde{h}^{g,L}_{n_1,n_2}\big\Vert^2 = \mathcal{T}^L_{n_1,n_2;n_1,n_2}\big(\widetilde{g}^{\dag}\widetilde{g}\big)= \mathcal{T}^L_{n_1,n_2;n_1,n_2}\big(\big(g^{\dag}g\big)^{-1}\big).
\end{gather}
Hence we are led to studying the asymptotic behavior of the expressions arising from~\eqref{matelgamapp} and~\eqref{diagmatel} respectively
 \begin{gather}
\label{aspolJacobi}
\mathcal{T}^L_{n_1,n_2;n_1,n_2}(h) = (\det h)^{n_1}  h_{22}^{n_2-n_1}  P_{n_1}^{(0,n_2-n_1)}\left(1 + 2\frac{\vert h_{12}\vert^2}{\det h}\right)\\
\hphantom{\mathcal{T}^L_{n_1,n_2;n_1,n_2}(h) =}{}
= (\det h)^{n_1}  h_{22}^{n_2-n_1}\, {}_2F_1\left(-n_1,n_2+1;1;\frac{\vert h_{12}\vert^2}{\det h}\right)  ,\nonumber
\end{gather}
where $h$ is positive and Hermitian.

Note that the alternative forms of this expression, obtained by exploiting the symmetry with respect to the interchange $1\rightarrow 2$, may be easer to manipulate
\begin{gather}
\mathcal{T}^L_{n_1,n_2;n_1,n_2}(h) =(\det h)^{n_2}  h_{11}^{n_1-n_2}\,{}_2F_1\left(n_1 + 1,- n_2;1;-\frac{\vert h_{12}\vert^2}{\det h}\right)
\nonumber\\
\label{alterform3}
\hphantom{\mathcal{T}^L_{n_1,n_2;n_1,n_2}(h)}{}
=  h_{11}^{n_1}  h_{22}^{n_2}\, {}_2F_1\left(-n_1, -n_2;1;\frac{\vert h_{12}\vert^2}{h_{11} h_{22}}\right)\\
\hphantom{\mathcal{T}^L_{n_1,n_2;n_1,n_2}(h)}{}
= (\det h)^{n_1+n_2 +1}  h_{11}^{-n_2-1}  h_{22}^{-n_1-1}\, {}_2F_1\left(n_1 + 1, n_2 +1;1;\frac{\vert h_{12}\vert^2}{h_{11} h_{22}}\right)  .
\nonumber
\end{gather}
The most symmetrical and simplest form is clearly \eqref{alterform3}, which, once expanded, reads as
\begin{gather}
\label{alterform3dev}
\mathcal{T}^L_{n_1,n_2;n_1,n_2}(h)=  h_{11}^{n_1} h_{22}^{n_2} \sum_{m=0}^{n_1\curlyvee n_2}\binom{n_1}{m}\binom{n_2}{m}  \left(-\frac{\vert h_{12}\vert^2}{h_{11} h_{22}}\right)^{m}  .
\end{gather}
From this expression, from the fact that for any positive hermitian matrix $\vert h_{12}\vert^2 \leq h_{11} h_{22}$ (strict inequality if the matrix is nonsingular), and from the well-known summation formula (e.g., see~\cite{grad})
\begin{gather*}
\sum_m \binom{n_1}{m}\binom{n_2}{p-m} = \binom{n_1+n_2}{p}  ,
\end{gather*}
we easily derive the following upper bound (keeping in mind $L=n_1+n_2$),{\samepage
\begin{gather}
\label{upbound}
\mathcal{T}^L_{n_1,n_2;n_1,n_2}(h)  \leq
    \binom{L}{n_1}   h_{11}^{n_1}  h_{22}^{n_2}  \leq (\operatorname{tr} h)^L  .
\end{gather}
From this follow the upper bounds for the norms of the  vectors in question given in~\eqref{upboundphipsi}.}

Next, we prove in Appendix~\ref{asympTs} the following estimates  of the diagonal matrix elements $\mathcal{T}^L_{n_1,n_2;n_1,n_2}(h) $:
\begin{gather}
\label{lowupboundMT}
 \frac{h_{11}^{n_1}  h_{22}^{n_2}}{\sqrt{\pi \min(n_1,n_2)}}\leq \mathcal{T}^L_{n_1,n_2;n_1,n_2}(h) \leq \frac{L!}{n_1! n_2!}   h_{11}^{n_1}  h_{22}^{n_2}  \leq (\operatorname{tr} h)^{L}  ,
\end{gather}
the lower bound being asymptotic at large $n_1$, $n_2$,  whereas the upper bound is valid for any~$n_1$,~$n_2$.
The application of these estimates to the norms of the polynomials $\{h^{g,L}_{n_1,n_2}\}$ and $\{h^{\widetilde g,L}_{n_1,n_2}\}$ (with $L=n_1+n_2$) given by~\eqref{normphi} and~\eqref{normpsi} yields the  inequalities~\eqref{lowupboundMT1} and~\eqref{lowupboundMT2}.
\end{proof}

As we have seen previously, the operators and the vectors introduced so far satisfy Assumptions~${\mathcal D}$-pb~\ref{Assumption1} and ${\mathcal D}$-pb~\ref{Assumption2}. 
On the other hand, the estimates above suggest that, using the explicit representation of our vectors in terms of our deformed complex Hermite polynomials, Assumption ${\mathcal D}$-pb~\ref{Assumption3} might be not satisf\/ied. In fact, in order for ${\mathcal F}_g:=\{h^{g,L}_{n_1,n_2}\}$ and ${\mathcal F}_{\tilde g}:=\{\widetilde{h}^{g,L}_{n_1,n_2}\}$ (or equivalently $\{\widetilde{H}^g_n ({\mathfrak{z}})\}$ and $\{\widetilde{H}^g_n ({\mathfrak{z}})\}$) to be bases for ${\mathcal H}$, the product of their norms should be bounded in~$n_1$ and~$n_2$, see, for instance, \cite[Lemma~3.3.3]{dav}. Now we derive from~\eqref{lowupboundMT1} and~\eqref{lowupboundMT2}
\begin{gather*}
 \frac{1}{\pi \min(n_1,n_2)}\left(\frac{ad}{\vert \det g\vert^{2}}\right)^L\leq \big\Vert h^{g,L}_{n_1,n_2}\big\Vert^2 \big\Vert \widetilde{h}^{g,L}_{n_1,n_2}\big\Vert^2 \leq  \binom{L}{n_1}^2 \left(\frac{ad}{\vert \det g\vert^{2}}\right)^L  .
\end{gather*}
 So unless $g$ is diagonal, we see from $ad > \vert \det g\vert^2$   that the product is  not bounded in~$L$.

However, it is possible to check  that ${\mathcal F}_g$ and ${\mathcal F}_{\tilde g}$ are ${\mathcal G}$-quasi bases, with ${\mathcal G}:=D(T(g)^\dagger)\cap D(T^{-1}(g))$, which is dense in ${\mathcal H}$ since it contains the linear span of the original polyno\-mials~$h_{n_1,n_2}$. This is a consequence of the fact that the vectors in  ${\mathcal F}_g$ and ${\mathcal F}_{\tilde g}$ are the image, via $(T(g)^\dagger$ and~$(T^{-1}(g)$, of the~$h_{n_1,n_2}$'s. So, this is enough to conclude that we are fully within the general pseudo-bosonic framework.

\section{Bi-displacement operators and bi-coherent states}
\label{bcpb}

We now consider how a pair $(a^{\sharp},b^{\sharp})$ of pseudo-bosonic operators, behaving as in Section~\ref{mathpb}, can be used to construct a generalized version of the canonical coherent states. Our analysis extends that originally contained in~\cite{bagpb1, tri}. First of all we introduce, at least formally, the two  $z$-dependent operators
\begin{gather*}
\mathfrak{D}(z)=\exp\{z b-\overline{z} a\}, \qquad \widetilde{\mathfrak{D}}(z)=\exp\big\{z a^\dagger-\overline{z} b^\dagger\big\}  .
\end{gather*}
They will be named bi-displacement operators, by analogy with the Weyl--Heisenberg case.

\begin{AssumptionW}\label{Assumption5}
 With the notations of Section \ref{mathpb}, for all $z\in{\mathbb C}$,  $\mathfrak{D}(z)$ and  $\widetilde{\mathfrak{D}}(z)$ are def\/ined in the dense subspace ${\mathcal D}$ of ${\mathcal H}$ such that $a^\sharp{\mathcal D}\subseteq{\mathcal D}$ and $b^\sharp{\mathcal D}\subseteq{\mathcal D}$.
\end{AssumptionW}

The  Weyl  formula, $e^A \,e^B = e^{\left( \frac{1}{2} [A,B]\right)} \, e^{(A + B)}$, (arising from the Baker--Campbell--Hausdorf\/f relation) which is valid for any pair of operators that  commute with their commutator, $ [A,[A,B]] = 0 =  [B,[A,B]] $, yields the following alternative, factorized forms of these operators
\begin{gather}
\label{facMD}
    {\mathfrak{D}} (z) =    e^{-\vert z\vert^2/2}  e^{z b}  e^{-\bar z  a}= e^{\vert z\vert^2/2} e^{-\bar z  a} e^{z b},\\
  \label{facTMD}  \widetilde{\mathfrak{D}}(z)=   e^{-\vert z\vert^2/2}  e^{z a^{\dag}}  e^{-\bar z b^{\dag}}= e^{\vert z\vert^2/2} e^{-\bar z b^{\dag}} e^{z  a^{\dag}}  .
\end{gather}

The operator-valued maps $z \mapsto {\mathfrak{D}}(z)$  and $z \mapsto \widetilde{\mathfrak{D}}(z)$ are (possibly local) projective  representations of  the abelian group  of translations of the  complex plane.
Indeed, let us apply the Weyl formula to the product ${\mathfrak{D}}(z_1) {\mathfrak{D}}(z_2)$.  One gets the composition rules
\begin{gather}
\label{compMD}
{\mathfrak{D}}(z_1)  {\mathfrak{D}}(z_2)  = e^{-i z_1 \wedge z_2} {\mathfrak{D}}(z_1 + z_2)  ,\\
\label{compTMD}
\widetilde{\mathfrak{D}}(z_1)  \widetilde{\mathfrak{D}}(z_2)  = e^{-i z_1 \wedge z_2} \widetilde{\mathfrak{D}}(z_1 + z_2)  ,
\end{gather}
where $z_1 \wedge z_2:= x_1 y_2 - x_2 y_1$ for $x_i= \operatorname{Re} z_i$, $y_i= \operatorname{Im} z_i$. In particular, since
${\mathfrak{D}}(0) = I= \widetilde{\mathfrak{D}}(0)$,
\begin{gather*}
({\mathfrak{D}}(z))^{-1}= {\mathfrak{D}}(-z) , \qquad \big(\widetilde{\mathfrak{D}}(z)\big)^{-1}= \widetilde{\mathfrak{D}}(-z)  .
\end{gather*}
These relations also give
\begin{gather}
\label{MDTMD}
 {\mathfrak{D}}(-z)= {\mathfrak{D}}(z)^{-1}= \widetilde{\mathfrak{D}}^\dagger(z)  .
\end{gather}
Let $\varpi(z) $ be a function on the complex plane obeying the (normalization) condition
\begin{gather}
\label{varpi0}
\varpi(0) = 1 ,
\end{gather}
and being assumed to def\/ine the two bounded operators ${\sf M}$ and $\widetilde{\sf M}$ on $\mathcal{H}$ through the operator-valued integrals
\begin{gather*}
{\sf M}  = \int_{\mathbb{C}} \varpi(z) \mathfrak{D}(z)   \frac{\mathrm{d}^2 z}{\pi}  ,\\
\widetilde{\sf M}  = \int_{\mathbb{C}} \varpi(z) \widetilde{\mathfrak{D}}(z)   \frac{\mathrm{d}^2 z}{\pi}= \int_{\mathbb{C}} \varpi(-z) \mathfrak{D}^{\dag}(z)   \frac{\mathrm{d}^2 z}{\pi}   .
\end{gather*}
Note that if we explicitly express the dependence of ${\mathsf{M}}$ on  the weight function, ${\mathsf{M}} \equiv {\mathsf{M}}^{\varpi}$, then $\widetilde{\mathsf{M}} \equiv  \left({\mathsf{M}}^{{\sf P}{\varpi}}\right)^{\dag}$, where ${\sf P}$ is the parity operator, ${\sf P}f(z) = f(-z)$. Hence, we have the interesting relation
\begin{gather*}
\varpi(z) = \overline{\varpi(-z)}\quad \forall\, z \quad \Rightarrow \quad {\sf M}^{\dag}= \widetilde{\sf M}  .
\end{gather*}
We now give the following proposition, where the fact that ${\mathfrak{D}}(z)$ and $\widetilde{\mathfrak{D}}(z)$ are def\/ined for each $z\in{\mathbb C}$ is crucial:

\begin{prop}\label{propcs}
If ${\mathfrak{D}}(z)$,  $\widetilde{\mathfrak{D}}(z)$, and $\varpi(z)$, are such that
\begin{gather}
\label{FubMz}
{\mathfrak{D}}(z)  \left[ \int_{\mathbb{C}}   {\mathfrak{D}}(z^{\prime})  \varpi(z^{\prime}) \frac{\mathrm{d}^2 z^{\prime}}{\pi}\right] {\mathfrak{D}}(-z)  =  \int_{\mathbb{C}}    {\mathfrak{D}}(z) {\mathfrak{D}}(z^{\prime})  {\mathfrak{D}}(-z)  \varpi(z^{\prime}) \frac{\mathrm{d}^2 z^{\prime}}{\pi}  , \\
\nonumber
\widetilde{\mathfrak{D}}(z)   \left[\int_{\mathbb{C}}   \widetilde{\mathfrak{D}}(z^{\prime})  \varpi(z^{\prime}) \frac{\mathrm{d}^2 z^{\prime}}{\pi}\right]  \widetilde{\mathfrak{D}}(-z)  =  \int_{\mathbb{C}}    \widetilde{\mathfrak{D}}(z) \widetilde{\mathfrak{D}}(z^{\prime})  \widetilde{\mathfrak{D}}(-z)  \varpi(z^{\prime}) \frac{\mathrm{d}^2 z^{\prime}}{\pi}
\end{gather}
hold,  for all $z$,  in a weak sense on the dense subspace ${\mathcal D}$ of ${\mathcal H}$,
then the families
\begin{gather}
\label{bidisMD}
   {\sf M}(z):=  {\mathfrak{D}}(z) {\sf M}{\mathfrak{D}}(-z)= {\mathfrak{D}}(z){\sf M}\widetilde{\mathfrak{D}}^{\dag}(z),\\
\nonumber
\widetilde{\sf M}(z):=  \widetilde{\mathfrak{D}}(z) \widetilde{\sf M}\widetilde{\mathfrak{D}}(-z)= \widetilde{\mathfrak{D}}(z){\sf M}{\mathfrak{D}}^{\dag}(z)
\end{gather}
 of bi-displaced operators under the  respective actions of  ${\mathfrak{D}}(z)$ and $\widetilde{\mathfrak{D}}(z)$ resolve the identity
 in the sense given in \eqref{resunipb}
\begin{gather}
\label{residMz}
\int_{\mathbb{C}}   {\sf M}(z)  \frac{\mathrm{d}^2 z}{\pi}= I  ,\\
\label{residTMz} \int_{\mathbb{C}}   \widetilde{\sf M}(z)  \frac{\mathrm{d}^2 z}{\pi}= I .
\end{gather}
\end{prop}

\begin{proof}
With the assumption \eqref{FubMz}, we have after applying  \eqref{compMD} twice
\begin{gather*}
\int_{\mathbb{C}}   {\sf M}(z)  \frac{\mathrm{d}^2 z}{\pi} = \int_{\mathbb{C}}  \frac{\mathrm{d}^2 z}{\pi}  \int_{\mathbb{C}}    {\mathfrak{D}}(z) {\mathfrak{D}}(z^{\prime})  {\mathfrak{D}}(-z) \varpi(z^{\prime}) \frac{\mathrm{d}^2 z^{\prime}}{\pi} \\
\hphantom{\int_{\mathbb{C}}   {\sf M}(z)  \frac{\mathrm{d}^2 z}{\pi}}{}
=  \int_{\mathbb{C}}  \frac{\mathrm{d}^2 z}{\pi} \int_{\mathbb{C}}  e^{-2i z\wedge z^{\prime}}    {\mathfrak{D}}(z^{\prime}) \varpi(z^{\prime}) \frac{\mathrm{d}^2 z^{\prime}}{\pi}  .
\end{gather*}
Then \eqref{residMz} is a direct consequence  of the formula (symplectic Fourier transform of the func\-tion~$1$ in the plane)
 \begin{gather}
\label{sympFour1}
\int_{\mathbb{C}} e^{-2i z\wedge z^{\prime}}     \frac{\mathrm{d}^2 z}{\pi} =\int_{\mathbb{C}}  e^{ z \bar z^{\prime} -\bar z z^{\prime}}    \frac{\mathrm{d}^2 z}{\pi}  = \pi \delta^{2}(z^{\prime})  ,
\end{gather}
and of the condition (\ref{varpi0}) with ${\mathfrak{D}}(0)= I$. The same demonstration applies trivially to \eqref{residTMz}.
\end{proof}

Let us expand the operators ${\mathfrak{D}}$ and $\widetilde{\mathfrak{D}}$ in terms of the biorthonormal bases or sets\footnote{Their nature depends on which one of the Assumptions ${\mathcal D}$-pb~\ref{Assumption3} or ${\mathcal D}$-pbw~\ref{Assumption4} is satisf\/ied.} \eqref{FB32},
\begin{gather}
\label{expMD}
  {\mathfrak{D}}(z)  =  \sum_{m, n}  {\mathfrak{D}}_{mn}(z) |\varphi_m\rangle\langle\Psi_n|  , \\
\nonumber
  \widetilde{\mathfrak{D}}(z)=   \sum_{m, n}  \widetilde{\mathfrak{D}}_{mn}(z) |\Psi_m\rangle\langle\varphi_n|  ,
\end{gather}
where the matrix elements  in~\eqref{expMD}, which involve associated Laguerre polynomials $L^{(\alpha)}_n(t)$, are calculated in Appendix \ref{AmatelMD}
\begin{gather}
\label{matelMD}
{\mathfrak{D}}_{mn}(z) := \langle \Psi_m|{\mathfrak{D}}(z)|\varphi_n\rangle  =  \sqrt{\dfrac{n!}{m!}} e^{-\vert z\vert^{2}/2}  z^{m-n}   L_n^{(m-n)}(\vert z\vert^{2})  ,   \qquad \mbox{for} \ m\geq n  ,
\end{gather}
with $L_n^{(m-n)}(t) = \frac{m!}{n!} (-t)^{n-m}L_m^{(n-m)}(t)$ for $n\geq m$. From \eqref{MDTMD} we have
\begin{gather}
\label{matelTMD}
 \widetilde{\mathfrak{D}}_{mn}(z)= \overline{ {\mathfrak{D}}}_{nm}(-z) .
\end{gather}
Note that the polynomial parts of these matrix elements are, up to a factor, complex Hermite polynomials.

As an interesting example, which is  inspired from \cite{cahillglauber69} (see also~\cite{bergaz14}), we choose
\begin{gather*}
\varpi_s(z) = e^{s |z|^2/2}  , \qquad \operatorname{Re}  s<1  .
\end{gather*}
Since this function is isotropic in the complex plane, the resulting operator ${\sf M}\equiv {\sf M}_s$ is diagonal.
when expanded in terms of the $\varphi_n$'s and $\Psi_n$'s in~\eqref{FB32}.  From the expression~(\ref{matelMD}) of the matrix elements of~${\mathfrak{D}}(z)$, and the integral~\cite{magnus66}
\begin{gather*}
\int_{0}^{\infty}   e^{-\nu x} x^{\lambda}  L_{n}^{\alpha}(x)\mathrm{d} x =\frac{\Gamma(\lambda + 1)\Gamma(\alpha+n+1)}{n!  \Gamma(\alpha + 1)}\nu^{-\lambda -1 }{}_{2}F_1(-n,\lambda + 1; \alpha + 1; \nu^{-1})  ,
\end{gather*}
we get the diagonal elements of  ${\sf M}_s$
\begin{gather*}
\langle \varphi_n|{\sf M}_s|\Psi_n\rangle = \frac{2}{1-s} \left( \frac{s+1}{s-1} \right)^n ,
\end{gather*}
and so
\begin{gather*}
{\sf M}_s= \int_{\mathbb{C}} \varpi_s(z) {\mathfrak{D}}(z)  \frac{{\mathrm{d}}^2z}{\pi }= \frac{2}{1-s} \exp \left\lbrack\left( \log \dfrac{s+1}{s-1}\right) ba \right\rbrack.
\end{gather*}
Then $s=-1$ corresponds to the basic operator
\begin{gather*}
{\sf M}_{-1}= \lim_{s\to -1_{-}} \dfrac{2}{1-s} \exp \left\lbrack\left( \ln \dfrac{s+1}{s-1}\right) b a \right\rbrack = |\varphi_0\rangle\langle \Psi_0| ,
\end{gather*}
\emph{Bi-coherent} states show up when precisely this operator is  bi-displaced along \eqref{bidisMD}
\begin{gather*}
{\sf M}_1(z):={\mathfrak{D}}(z) {\sf M}_1{\mathfrak{D}}(-z)= {\mathfrak{D}}(z) |\varphi_0\rangle\langle \Psi_0| \widetilde{\mathfrak{D}}^{\dag}(z)\equiv |\varphi(z)\rangle\langle \Psi(z)| ,
\end{gather*}
i.e., they are def\/ined as
\begin{gather}
\varphi(z)={\mathfrak{D}}(z)\varphi_0,\qquad \Psi(z)=\widetilde{\mathfrak{D}}(z) \Psi_0 .
\label{92}
\end{gather}
It is necessary to check that these vectors are well def\/ined in ${\mathcal H}$ for some $z\in{\mathbb C}$. Using the factorizations~\eqref{facMD} and~\eqref{facTMD} together with the properties of~$\varphi_0$ and~$\Psi_0$, we get
\begin{gather}
\varphi(z)=e^{-|z|^2/2}\sum_{n=0}^\infty\frac{z^n}{\sqrt{n!}}\varphi_n , \qquad
\Psi(z)=e^{-|z|^2/2}\sum_{n=0}^\infty\frac{z^n}{\sqrt{n!}}\Psi_n .
\label{93}
\end{gather}
Since ${\mathfrak{D}}(z)$ and $\widetilde{\mathfrak{D}}(z)$ are not unitary operators, or alternatively since $\varphi_n$ and $\Psi_n$ are not normalized in general, we should concretely check that the series in~\eqref{93} both converge, at least for some reasonably large set of~$z$'s. In fact, so far we have assumed that the states exist for all $z\in{\mathbb C}$. This is clear whenever~${\mathcal F}_\varphi$ and~${\mathcal F}_\Psi$ are o.n.\ bases (in this case, in fact, convergence is for all $z\in{\mathbb C}$), or when $\varphi_0\in D({\mathfrak{D}}(z))$ and $\Psi_0\in D(\widetilde{\mathfrak{D}}(z))$, but it is not evident in general. However, it is possible to prove the following.

\begin{prop}
Suppose that $r_\varphi, r_\Psi>0$ and $0\leq \alpha_\varphi,\alpha_\Psi<\frac{1}{2}$ exist such that
\begin{gather}
\|\varphi_n\|\leq r_\varphi^n(n!)^{\alpha_\varphi},\qquad \|\Psi_n\|\leq r_\Psi^n(n!)^{\alpha_\Psi},
\label{93bis}
\end{gather}
then $\varphi(z)$ is well defined for all $z$.
\end{prop}

\begin{proof}
The proof relies upon the following estimate
\begin{gather*}
\|\varphi(z)\|^2=e^{-|z|^2} \sum_{n,k=0}^\infty \frac{z^n\overline{z}^k}{\sqrt{n!k!}} \langle \varphi_k,\varphi_n\rangle \leq e^{-|z|^2} \left(\sum_{n,k=0}^\infty\frac{(r_\varphi |z|)^n}{(n!)^{\frac{1}{2}-\alpha_\varphi}}\right)^2,
\end{gather*}
which converges for all values of $z\in{\mathbb C}$
\end{proof}

Analogously, we can prove that $\Psi(z)$ is well def\/ined for all $z$. Moreover, $\langle \varphi(z),\Psi(z)\rangle =\langle \varphi_0,\Psi_0\rangle =1$.

 Notice that the inequalities in~(\ref{93bis}) are surely satisf\/ied for Riesz bases, since in this case the norms of both $\varphi_n$ and $\Psi_n$ are uniformly bounded in~$n$. However, our assumption here does not prevent us from considering families ${\mathcal F}_\varphi$ and ${\mathcal F}_\Psi$ of vectors with divergent norms, as often happen in explicit models~\cite{bagbook}. In other words,  $\varphi(z)$ and~$\Psi(z)$ could also be def\/ined  if ${\mathcal F}_\varphi$ and ${\mathcal F}_\Psi$ are not bases, which happens if both  $\|\varphi_n\|$ and $\|\Psi_n\|$ diverge with~$n$~\cite{bagbook}, at least if condition~(\ref{93bis}) holds true.

It is interesting to notice that conditions (\ref{93bis}) are indeed satisf\/ied in several models recently considered in the literature. For instance, in~\cite{bag2}, the vector $\varphi_n$ satisf\/ies an inequality like $\|\varphi_n\|\leq(1+|\alpha-\beta|)^{n/2}$, where $\alpha$ and $\beta$ are two in general complex parameters of the model. A~similar estimate, with a harmless overall constant, can also be found for~$\|\Psi_n\|$. This means that~(\ref{93bis}) also are satisf\/ied in the model originally proposed in~\cite{bag4}, which is a special case of that in \cite{bag2}, and for the model discussed in~\cite{bag3}, which is a two-dimensional, non commutative, version of the same model.

Also the vectors introduced in the Swanson model~\cite{bag4} satisfy similar inequalities. Indeed~\cite{bagbook}, in this model we have found that
\begin{gather*}
\|\varphi_n\|^2=|N_1|^2\sqrt{\frac{\pi}{\cos(2\theta)}} P_n\left(\frac{1}{\cos(2\theta)}\right),\qquad \|\Psi_n\|=\left|\frac{N_2}{N_1}\right|\|\varphi_n\|,
\end{gather*}
where $N_1$ and $N_2$ are normalization constants, $\theta$ is a parameter in $\big]{-}\frac{\pi}{4},\frac{\pi}{4}\big[$, and $P_n$ is a Legendre polynomial. Using \cite{szego} we deduce that, for instance,
\begin{gather*}
\|\varphi_n\|\leq A_\theta \alpha_\theta^n,
\end{gather*}
where $A_\theta$ is a constant and
\begin{gather*}
\alpha_\theta=\sqrt{\frac{1}{\cos(2\theta)}+\left(\frac{1}{\cos^2(2\theta)}-1\right)^{1/2}}.
\end{gather*}

\begin{Remark}
 The possibility remains open that $\varphi(z)$ and $\Psi(z)$ only exist for $z\in{\mathcal E}$, with~${\mathcal E}$ a~proper (suf\/f\/iciently large) subset of~${\mathbb C}$. When this is so, of course, the proof of Proposition~\ref{propcs} fails to work since the integral in~(\ref{sympFour1}) will only be extended to ${\mathcal E}$ and not to all ${\mathbb C}$. Therefore, our bi-coherent states need not to resolve the identity anymore.  The analysis of this situation is postponed to another paper.
\end{Remark}

Another reason why these vectors are called bi-coherent is because they are eigenstates of some lowering operators. Indeed we can check that
\begin{gather*}
a\varphi(z)=z\varphi(z), \qquad b^\dagger\Psi(z)=z\Psi(z),
\end{gather*}
for all $z\in{\mathbb C}$. Their overlap is given by the kernel
\begin{gather*}
\mathcal{K}(z,z^{\prime})=\langle \varphi(z)|\Psi(z^{\prime})\rangle= e^{-|z|^2/2} e^{-|z^{\prime}|^2/2}  e^{\bar z z^{\prime}}=  e^{-|z-z^{\prime}|^2/2}  e^{i z \wedge z^{\prime}} ,
\end{gather*}
which is the same kernel as that for the canonical coherent states.
As a particular case of~\eqref{residMz}, they resolve the identity
\begin{gather}
\label{bicoresI}
\int_{\mathbb{C}}   | \varphi(z)\rangle\langle\Psi(z)|   \frac{\mathrm{d}^2 z}{\pi}= I = \int_{\mathbb{C}}   | \Psi(z)\rangle\langle\varphi(z)|   \frac{\mathrm{d}^2 z}{\pi} ,
\end{gather}
and this entails the reproducing property of the kernel
\begin{gather*}
\int_{\mathbb{C}}   \mathcal{K}(z,z^{\prime}) \mathcal{K}(z^{\prime},z^{\prime\prime})\frac{\mathrm{d}^2 z^{\prime}}{\pi} = \mathcal{K}(z,z^{\prime\prime}).
\end{gather*}
Finally note the projective covariance property of bi-coherent states, as a direct consequence of~\eqref{92} and~\eqref{compMD} and~\eqref{compTMD}
\begin{gather}
\label{covbicoh}
{\mathfrak{D}}(z)\varphi(z^{\prime}) =e^{-i z \wedge z^{\prime}} \varphi(z+ z^{\prime})   , \qquad \widetilde{\mathfrak{D}}(z)\Psi(z^{\prime}) = e^{-i z \wedge z^{\prime}}\Psi(z+ z^{\prime}) .
\end{gather}

\section{Integral quantization with bicoherent families and more}
\label{quantpb}

We now adapt the integral quantization scheme described in \cite{aagbook13,bergaz14} and \cite{balfrega14} to the present pseudo-bosonic formalism. This is made possible when the resolutions of the identity~\eqref{residMz} and~\eqref{residTMz} are valid on some dense subspace of the Hilbert space in question. Given a weight function $\varpi(z)$ with $\varpi(0)=1$ and the resulting families of bi-displaced operators~${\mathsf{M}}(z)$ and~$\widetilde{\mathsf{M}}(z)$, the quantizations of a function $f(z)$ on the complex plane  is def\/ined by the linear maps
\begin{gather}
\label{qfM} f\mapsto A_{f}   =\int_{\mathbb{C}}f(z)  {\sf M}(z) \frac{\mathrm{d}^{2}z}{\pi}=\int_{\mathbb{C}}\mathcal{F}(-z)  {\mathfrak{D}}(z) \varpi(z)\frac{\mathrm{d}^{2}z}{\pi} , \\
\nonumber
f\mapsto \widetilde A_{f}   =\int_{\mathbb{C}}f(z) \widetilde{\mathsf{M}}(z) \frac{\mathrm{d}^{2}z}{\pi}=\int_{\mathbb{C}}\mathcal{F}(-z)  \widetilde{\mathfrak{D}}(z) \varpi(z)\frac{\mathrm{d}^{2}z}{\pi} ,
\end{gather}
where $\mathcal{F}$ is the symplectic Fourier transform of $f$,
\begin{gather*}
 \mathcal{F}[f](z)\equiv \mathcal{F}(z)=\int_{\mathbb{C}}f(\xi)  e^{z\bar{\xi}-\bar{z}\xi} \frac{\mathrm{d}^{2}\xi}{\pi}= \int_{\mathbb{C}}f(\xi)  e^{2i\xi\wedge z} \frac{\mathrm{d}^{2}\xi}{\pi} .
\end{gather*}
Covariance with respect to translations reads $A_{f(z-z_{0})}={\mathfrak{D}}(z_{0})A_{f(z)}{\mathfrak{D}}(z_{0})^{\dagger}$. In the case of a~real even weight function we have the relation $\widetilde A_{f}= A_{\bar f}^{\dag}$, and then, if the function $f$ is real, the adjoint of~$A_f$ is~$\widetilde A_{f}$. A more delicate question is to f\/ind pairs $(f,\varpi)$ for which $A_f$ is symmetric.

In the sequel we focus on the quantizations using~${\mathsf{M}}(z)$ only, since there are well-def\/ined relations between~$A_{f}$ and~$\widetilde A_{f}$.

We now show that the generic pseudo-boson commutation rule~\eqref{FB31} is always the outcome of the above quantization, whatever the chosen complex function~$\varpi (z )$, provided
integrability and dif\/ferentiability at the origin is ensured. For this let us calculate
$A_{z}$ and $A_{\bar{z}}$. Taking into account that the symplectic Fourier transform of the function $z$, $\mathcal{F}[z](z)$, is equal to $-\partial_{\bar{z}}\pi\delta^{2}(z)$,
where $\pi\delta^{2}(z)=\int_{\mathbb{C}}e^{z\bar{\xi}-\bar{z}\xi}\frac{d^{2}\xi}{\pi}$,
one has from~\eqref{qfM} $A_{z}=-\left[\varpi(z)\partial_{\bar{z}}{\mathfrak{D}}(z)+{\mathfrak{D}}(z)\partial_{\bar{z}}\varpi(z)\right]_{z=0}$.
Then, using $\partial_{\bar{z}}{\mathfrak{D}}(z)=-\left(a-\frac{z}{2}\right){\mathfrak{D}}(z)$
we obtain f\/inally
\begin{gather*}
A_{z}=a\varpi (0 )-\left.\partial_{\bar{z}}\varpi\right\vert _{z=0}= a-\left.\partial_{\bar{z}}\varpi\right\vert _{z=0} .
\end{gather*}
Similarly, we obtain for $A_{\bar{z}}$ the following expression
\begin{gather*}
A_{\bar{z}}=b+\left.\partial_{z}\varpi\right\vert _{z=0} ,
\end{gather*}
after using the relation $\partial_{z}{\mathfrak{D}}(z)=\left(b-\frac{\bar{z}}{2}\right){\mathfrak{D}}(z)$.

Def\/ining the  Poisson bracket for functions $f(z)$ (actually $f(z,\bar z))$ as
\begin{gather*}
\{f,g\}= \frac{\partial f}{\partial z}   \frac{\partial g}{\partial \bar z} - \frac{\partial g}{\partial z} \frac{\partial f}{\partial \bar z}  ,
\end{gather*}
we thus check that the map \eqref{qfM} is ``pseudo-canonical'' in the sense that
\begin{gather*}
\{z,\bar z\} = 1 \mapsto [a,b]= I  .
\end{gather*}

\section{Conclusion}
\label{conclu}

In this paper we have discussed two illustrations of the ${\mathcal D}$-pseudo-bosonic formalism, the biortho\-go\-nal complex Hermite polynomials, and a second using families of vectors and operators in the underlying Hilbert space, built in way similar to that of the standard coherent states, i.e.,  as orbits of  projective representations of the Weyl--Heisenberg group. We have also considered the resolutions of the identity satisf\/ied by these families and the related integral quantizations naturally arising from them. In particular, these quantizations  yield  exactly the genuine pseudo-bosonic commutation rules.

These results can be of some interest in connection with PT or pseudo-hermitian quantum mechanics, where the role of self-adjoint operators is usually not so relevant. In~\cite{bagbook} several connections have been already established between  ${\mathcal D}$-pseudo-bosons and this {\em extended} quantum mechanics, and our results on complex Hermite polynomials and {on dual integral quantizations} suggest that more can be established. This is, in fact, part of our future work.

\appendix

\section[Irreducible finite-dimensional representations of ${\rm GL}(2,{\mathbb C})$]{Irreducible f\/inite-dimensional representations of $\boldsymbol{{\rm GL}(2,{\mathbb C})}$}
\label{appuir}

The linear action of a  $2\times2$ complex matrix $\gamma = \begin{pmatrix}
  \gamma_{11}    &   \gamma_{12}  \\
  \gamma_{21}     &  \gamma_{22}
\end{pmatrix}$ on vector space ${\mathbb C}^2$ is def\/ined in the usual way as
\begin{gather*}
\gamma\cdot\mathbf{x}= \gamma\binom{x_1}{x_2} = \binom{ \gamma_{11} x_1+  \gamma_{12}x_2}{ \gamma_{21}x_1 +   \gamma_{22}x_2} .
\end{gather*}
We now consider the linear representation $\mathcal{T}^s$ of ${\rm GL}(2, {\mathbb C})$ carried on by the complex vector space $\mathcal{V}^s$ of two-variable homogeneous polynomials $p(\mathbf{x})$ of f\/ixed degree $s$ in the following way
\begin{gather}
\label{repglpol}
\left(\mathcal{T}^s(\gamma)p\right)(\mathbf{x})= p\big({}^t\gamma\cdot\mathbf{x}\big) ,
\end{gather}
where ${}^t\gamma$ is the transpose of $\gamma$. In the monomial basis of f\/ixed degree~$s$
\begin{gather}
\label{monbas}
e_{n_1,n_2}(\mathbf{x}) := \frac{x_1^{n_1}x_2^{n_2}}{\sqrt{n_1! n_2!}}  , \qquad s=n_1+n_2
\end{gather}
the matrix elements $\mathcal{T}^s_{n'_1,n'_2;n_1,n_2}(\gamma)$ of $\mathcal{T}^s(\gamma)$, def\/ined by
\begin{gather}
\label{defmatel}
e_{n_1,n_2}({}^t\gamma\cdot\mathbf{x}) = \sum_{n'_1,n'_2=s-n^{\prime}_1}\mathcal{T}^s_{n'_1,n'_2;n_1,n_2}(\gamma)  e_{n'_1,n'_2}(\mathbf{x}) ,
\end{gather}
are given by
\begin{gather}
\nonumber \mathcal{T}^s_{n'_1,n'_2;n_1,n_2}(\gamma) =  \sqrt{\frac{n'_1!n'_2!}{n_1!n_2!}} \gamma_{21}^{n_1} \gamma_{12}^{n'_1} \gamma_{22}^{n_2-n'_1}\sum_{j=0}^{n'_1}\binom{n_1}{j}  \binom{n_2}{n'_1-j}  \rho^{j}\\
\hphantom{\mathcal{T}^s_{n'_1,n'_2;n_1,n_2}(\gamma)}{}
 =\sqrt{\frac{n'_1!n'_2!}{n_1!n_2!}} (\det \gamma)^{n'_1}  \gamma_{21}^{n_1-n'_1}  \gamma_{22}^{n_2-n'_1} P_{n'_1}^{(n_1-n'_1,n_2-n'_1)}\left(\frac{\rho+1}{\rho-1}\right)  , \label{matelgamapp}\\
  \rho:= \frac{\gamma_{11}\gamma_{22}}{\gamma_{12}\gamma_{21}}  .\nonumber
\end{gather}
We impose the constraints $n_1+n_2= s=n^{\prime}_1 + n^{\prime}_2$, which have to be satisf\/ied in all these expressions. However we keep the two summation indices for notational convenience.
The polynomials $P_{n}^{(\alpha,\beta)}$ are the Jacobi polynomials  given \cite{magnus66} by
\begin{gather*}
P_{n}^{(\alpha,\beta)}(x)  = 2^{-n}(x-1)^n\sum_{j=0}^n\binom{n+\alpha}{j} \binom{n+\beta}{n-j}\left(\frac{x+1}{x-1}\right)^j \\
\hphantom{P_{n}^{(\alpha,\beta)}(x)}{}
= \binom{n+\alpha}{n}\, {}_2F_1\left(-n,\alpha+\beta+n+1;\alpha+1;\frac{1-x}{2}\right)  .
\end{gather*}
Note the alternative and simpler  form of \eqref{matelgamapp} in terms of the above hypergeometric function, due to the relation $\frac{\rho+1}{\rho-1}= 1 +2\dfrac{\gamma_{12}\gamma_{21}}{\det \gamma}$:
\begin{gather*}
\mathcal{T}^s_{n'_1,n'_2;n_1,n_2}(\gamma)
= \sqrt{\frac{n'_1!n'_2!}{n_1!n_2!}} (\det \gamma)^{n'_1} \gamma_{21}^{n_1-n'_1} \gamma_{22}^{n_2-n'_1} \binom{n_1}{n'_1} \\
\hphantom{\mathcal{T}^s_{n'_1,n'_2;n_1,n_2}(\gamma)=}{}
\times {}_2F_1\left(-n'_1,n'_2+1;n_1-n'_1+1;\frac{\gamma_{12}\gamma_{21}}{\det \gamma}\right).
\end{gather*}
In particular the diagonal elements read as
\begin{gather}
\label{diagmatel}
\mathcal{T}^s_{n_1,n_2;n_1,n_2}(\gamma)= (\det \gamma)^{n_1} \gamma_{22}^{n_2-n_1}\, {}_2F_1\left(-n_1,n_2+1;1;\frac{\gamma_{12}\gamma_{21}}{\det \gamma}\right).
\end{gather}
Finally, note the property
\begin{gather}
\mathcal T^s\big(g^\dag\big) = \big(\mathcal T^s (g)\big)^*  .
\label{star-rep}
\end{gather}

\section{Asymptotic behavior of matrix elements}
\label{asympTs}
In this appendix we give the asymptotic behavior of the diagonal  matrix elements $\mathcal{T}^s_{n_1,n_2;n_1,n_2}(h)$, for a positive matrix $h$,  for large $n_1$, $n_2$, for two types of directions in the positive two-dimensional  square lattice  $\Lambda_{++}= \{(n_1,n_2)\, | \, n_1 , n_2 \in {\mathbb N}\}$.

\subsection*{Behavior at large $\boldsymbol{n_1}$, $\boldsymbol{n_2}$, with f\/ixed $\boldsymbol{d= n_2- n_1}$}

To study this behavior, we use the expression \eqref{aspolJacobi},  with $d=n_2-n_1$, of the diagonal elements in terms of the Jacobi polynomials:
\begin{gather*}
\mathcal{T}^s_{n_1,n_2;n_1,n_2}(h)= (\det h)^{n_1}  h_{22}^{d}  P_{n_1}^{(0,d)}\left(X\right)  , \qquad X= 1 + 2\frac{\vert h_{12}\vert^2}{\det h}  .
\end{gather*}
We suppose $d\geq 0$ with no  loss of  generality. From \cite{magnus66} we know that, at large $n$,
\begin{gather*}
P_n^{(\alpha,\beta)}(x)\sim \frac{1}{\sqrt{2\pi n}} (x-1)^{-\frac{\alpha}{2}}(x+1)^{-\frac{\beta}{2}} \!\big[\sqrt{x+1} + \sqrt{x-1}\big]^{\alpha + \beta} (x-1)^{-\frac{1}{4}} \!\big[ x + \sqrt{x^2-1}\big]^{n+ \frac{1}{2}}
\end{gather*}
holds for $x>1$ or $x<1$.
Applied to the present case this leads to the asymptotic behavior of $\mathcal{T}^s_{n_1,n_2;n_1,n_2}(h)$ for f\/ixed $d= n_2-n_1 >0$
\begin{gather}
\label{asympTsdI}
\mathcal{T}^s_{n_1,n_2;n_1,n_2}(h)  \sim  \frac{\left(2\vert h_{12}\vert^2  \det h\right)^{-1/4}}{\sqrt{2\pi n_1}}
 \left(\frac{h_{22}}{h_{11}}\right)^{d/2}  \left[ \sqrt{h_{11} h_{22}} + \vert h_{12}\vert\right]^{n_1 + n_2 +1}\\
 \hphantom{\mathcal{T}^s_{n_1,n_2;n_1,n_2}(h)}{}
\label{asympTsdII}  = \frac{h_{11}^{n_1}  h_{22}^{n_2}}{\sqrt{2\pi n_1}}  [2r(1-r)]^{-1/4}  \left[1 + \sqrt{r}\right]^{n_1+n_2 + 1}  , \qquad 0 < r= \frac{\vert h_{12}\vert^2}{h_{11}h_{22}} < 1 .
\end{gather}

\subsection*{Complete estimate  for large $\boldsymbol{n_1}$, $\boldsymbol{n_2}$}

Since the \textit{a priori} f\/ixed $d= n_2-n_1$ can be arbitrarily large, \eqref{asympTsdI} is valid for arbitrarily large $n_1$ and $n_2 \geq n_1$. In the case $n_1> n_2$ it is enough to  permute $1\leftrightarrow 2$ in the right-hand side of~\eqref{asympTsdII}.
This formula provides a lower bound to $\mathcal{T}^s_{n_1,n_2;n_1,n_2}(h)$ since  for any $r \in [0,1]$,  $0\leq r(1-r) \leq 1/4$, with maximum reached for $r=1/2$,  and $1\leq 1+ \sqrt r \leq 2$.

Hence, using also \eqref{upbound}, we get the estimates
\begin{gather*}
 \frac{h_{11}^{n_1} h_{22}^{n_2}}{\sqrt{\pi \min(n_1,n_2)}}\leq \mathcal{T}^s_{n_1,n_2;n_1,n_2}(h) \leq \frac{(n_1 + n_2)!}{n_1!n_2!}  h_{11}^{n_1} h_{22}^{n_2}  \leq (\mathrm{tr}\,h)^{n_1 + n_2} ,
\end{gather*}
the lower bound being asymptotic at large $n_1$, $n_2$,  whereas the upper bound is valid for any~$n_1$,~$n_2$.

\subsection*{Another exploration: behavior at large $\boldsymbol{n_1}$, $\boldsymbol{n_2}$, with f\/ixed $\boldsymbol{\nu = n_2/n_1}$}

Without loss of generality we suppose $\nu >1$.
In order to  analyze  this case, it is better to start from the simplest expression \eqref{alterform3dev}
\begin{gather}
\label{symsumA}
\mathcal{T}^s_{n_1,n_2;n_1,n_2}(h)=  h_{11}^{n_1}  h_{22}^{n_2} \sum_{m=0}^{n_1}\binom{n_1}{m}\binom{\nu n_1}{m}  r^{m}  , \qquad r= \frac{\vert h_{12}\vert^2}{h_{11} h_{22}}  .
\end{gather}
We know already that for $r=1$ this simplif\/ies to
\begin{gather}
\label{symsumr1}
\mathcal{T}^s_{n_1,n_2;n_1,n_2}(h)=  h_{11}^{n_1} h_{22}^{n_2} \binom{n_1+ n_2}{n_1}= h_{11}^{n_1}\, h_{22}^{n_2}\,\frac{(n_1 +n_2)!}{n_1! n_2!} ,
\end{gather}
From the Stirling formula,
\begin{gather*}
n!\sim\sqrt{2\pi}  e^{-n} n^{n+1/2}\qquad \mbox{for large $n$} ,
\end{gather*}
we see that, at large $n_1$, $n_2$, \eqref{symsumr1} behaves as
\begin{gather*}
\mathcal{T}^s_{n_1,n_2;n_1,n_2}(h) \sim \sqrt{\frac{n_1+n_2}{2\pi n_1 n_2}}  \frac{(n_1 +n_2)^{n_1+n_2}}{n_1^{n_1}  n_2^{n_2}}  h_{11}^{n_1}  h_{22}^{n_2}  ,
\end{gather*}
and if $n_2= \nu n_1$ with large~$\nu$, the above expression becomes
\begin{gather*}
\mathcal{T}^s_{n_1,n_2;n_1,n_2}(h) \sim \sqrt{\frac{n_1+n_2}{2\pi n_1 n_2}}  \frac{(n_1 +n_2)^{n_1+n_2}}{n_1^{n_1}  n_2^{n_2}}  h_{11}^{n_1}  h_{22}^{n_2}\sim \frac{1}{\sqrt{2\pi n_1 \nu}}  (1+\nu)^{n_1 + 1/2}  (e h_{11})^{n_1}  h_{22}^{\nu n_1}  .
\end{gather*}

We now consider the general case $r<1$. From the Stirling formula
we derive the asymptotic behavior of binomial coef\/f\/icient at large $n$,
\begin{gather*}
\nonumber \binom{n}{m=n\xi} \sim \frac{1}{\sqrt{2\pi n\xi(1-\xi)}}  e^{-n(\xi\log \xi + (1-\xi) \log(1-\xi))}  ,
\end{gather*}
where we have introduced the ``continuous'' variable  $\xi = m/n$, $0<\xi<1$. In the present case, we write $\xi = m/n_1$ and we  replace  the  sum $\sum\limits_{m=0}^{n_1}$ in~\eqref{symsumA} by the integral $\int_0^1 n_1\mathrm{d} \xi$ (or $\int_{\epsilon}^{1-\epsilon'} n_1\mathrm{d} \xi$ if some regularization is needed).  We obtain
\begin{gather*}
 \mathcal{T}^s_{n_1,n_2;n_1,n_2}(h)\sim  h_{11}^{n_1} h_{22}^{n_2} \frac{1}{2\pi}\int_{\epsilon}^{1-\epsilon'}\left[\xi^2(1-\xi) \left(1-\frac{\xi}{\nu}\right)\right]^{-1/2} e^{n_1A(\xi)} \mathrm{d} \xi,
\end{gather*}
with
\begin{gather*}
A(\xi)= -\left[ 2\xi\log \xi  - \xi\log \nu r +  (1-\xi)\log(1- \xi) + \nu\left(1-\frac{\xi}{\nu}\right)\log\left(1- \frac{\xi}{\nu}\right)\right] .
\end{gather*}
 Next we apply the Laplace method for evaluating the above integral for large~$n_1$ and~$\nu$, ignoring the divergence at the origin.  Laplace's approximation formula (with suitable conditions on the functions involved) reads
\begin{gather*}
\int_a^b h(x)  e^{nA(x)} \mathrm{d} x \approx \sqrt{\frac{2\pi}{n\vert A''(x_0)\vert}}h(x_0) e^{nA(x_0)}\qquad \mbox{as} \ n\to \infty ,
\end{gather*}
where $A'(x_0)= 0$ for $x_0\in [a,b]$, $A''(x_0)< 0$  and $h$ is positive.
Here, we have
\begin{gather*}
A'(\xi)= \log(1-\xi)\left(1-\frac{\xi}{\nu}\right) -\log \xi^2  + \log \nu r   ,
\\
A''(\xi) = - \frac{1}{1-\xi} -\frac{2}{\xi} - \frac{1}{\nu-\xi} .
\end{gather*}
We notice that $A'(\xi)< 0$ in the integration interval.
The equation $A'(\xi) = 0$ is equivalent to
\begin{gather*}
\xi^2  + \frac{r(1+\nu)}{1-r}\xi - \frac{r\nu}{1-r}= 0 .
\end{gather*}
The positive root is
\begin{gather*}
\xi_+= \frac{\sqrt{r}}{2(1-r)}\left[\sqrt{r(\nu-1)^2 +4\nu} - \sqrt{r}(1+\nu)\right] .
\end{gather*}
We easily check that $\xi_+= 0$ at $r=0$, that $\xi_+ \to \nu/(1+\nu)$ as $r\to 1_-$, and that, at f\/ixed $\nu \geq 1$, $\xi_+>0$ in the range $0<r \leq 1$. Also, for $\nu = 1$, $\xi_+= \sqrt{r}/(1+\sqrt{r})$.  Therefore, $\xi_+\in (0,1)$ for all $r \in (0,1)$ and $\nu\in [1,\infty)$.
Now, we have at $\xi= \xi_+$
\begin{gather*}
A (\xi_+ )= -\log(1-\xi_+)\left(1-\frac{\xi_+}{\nu}\right)^{\nu} ,
\\
A'' (\xi_+ )= -\frac{2-\left(1+\frac{1}{\nu}\right)\xi_+}{\xi_+  (1-\xi_+)  \left(1-\frac{\xi_+}{\nu}\right) }  .
\end{gather*}
Applying the Laplace formula yields the f\/inal result
\begin{gather*}
 \mathcal{T}^s_{n_1,n_2;n_1,n_2}(h)\sim  \frac{h_{11}^{n_1} h_{22}^{n_2}}{\sqrt{2\pi n_1}} \left[\xi_+ \left(2-\left(1 +\frac{1}{\nu}\right)\xi_+\right)\right]^{-1/2}(1-\xi_+)^{-n_1} \left(1-\frac{\xi_+}{\nu}\right)^{-n_2} .
\end{gather*}

\section[Matrix elements of ${\mathfrak{D}}(z)$ and $\widetilde{\mathfrak{D}}(z)$]{Matrix elements of $\boldsymbol{{\mathfrak{D}}(z)}$ and $\boldsymbol{\widetilde{\mathfrak{D}}(z)}$}
\label{AmatelMD}

The calculation of the matrix elements ${\mathfrak{D}}_{mn}(z):= \langle \Psi_m | {\mathfrak{D}}(z) | \varphi_n \rangle$ of ${\mathfrak{D}}(z)$ (and consequently for $\widetilde{\mathfrak{D}}(z)$ because of the relation~\eqref{matelTMD}) can be carried out by using the resolution of the identity~\eqref{bicoresI}, satisf\/ied by the bi-coherent states, their projective  covariance property~\eqref{covbicoh} and the reproducing properties of their overlap function
\begin{gather*}
  \langle \Psi_m | {\mathfrak{D}}(z) | \varphi_n \rangle  = \iint_{{\mathbb C}^2}\frac{d^2 z^{\prime}}{\pi}  \frac{d^2 z^{\prime\prime}}{\pi}\langle \Psi_m |\varphi(z^{\prime})\rangle \langle \Psi(z^{\prime})| {\mathfrak{D}}(z) | \varphi(z^{\prime\prime})\rangle \langle \Psi(z^{\prime\prime})|\varphi_n\rangle\\
\hphantom{\langle \Psi_m | {\mathfrak{D}}(z) | \varphi_n \rangle}{}
 =\frac{1}{\sqrt{m!n!}} \iint_{{\mathbb C}^2}\frac{d^2 z^{\prime}}{\pi}  \frac{d^2 z^{\prime\prime}}{\pi}  e^{-\frac{\vert z^{\prime} \vert^2}{2}}  e^{-\frac{\vert z^{\prime\prime} \vert^2}{2}}  {z^{\prime}}^m  {\bar{z^{\prime\prime}}}^n  e^{i\operatorname{Im} (z \bar{z^{\prime\prime}})}  \langle \Psi(z^{\prime})| \varphi(z^{\prime\prime} + z)\rangle\\
\hphantom{\langle \Psi_m | {\mathfrak{D}}(z) | \varphi_n \rangle}{}
= \frac{e^{-\frac{\vert z \vert ^2}{2}}}{\sqrt{m!n!}} \int_{{\mathbb C}}\frac{d^2 z^{\prime\prime}}{\pi}e^{-\vert z^{\prime\prime}\vert^2}  e^{-\bar{z z^{\prime}}} {\bar{z^{\prime\prime}}}^n (z^{\prime\prime} + z)^m .
\end{gather*}
After  binomial and exponential expansions and integration, one ends up with the following expression
\begin{gather*}
\langle \Psi_m | {\mathfrak{D}}(z) | \varphi_n \rangle  =
\begin{cases} \displaystyle \sqrt{\frac{n!}{m!}}  e^{-\frac{\vert z \vert ^2}{2}}  \alpha^{m-n}  L_{n}^{(m-n)}\left(\vert z\vert^2\right) & \mbox{for} \ \  n \leq m  , \vspace{1mm} \\
\displaystyle  \sqrt{\frac{m!}{n!}}  e^{-\frac{\vert z \vert ^2}{2}}  (-\bar{z})^{n-m}  L_{m}^{(n-m)}\left(\vert z \vert^2\right) & \mbox{for} \ \ n > m ,
\end{cases}
\end{gather*}
where
\begin{gather*}
L_n^{(\mu)}(x) = \sum_{k=0}^n (-1)^k \frac{\Gamma(n + \mu +1)}{\Gamma(\mu + k +1)(n-k)!}  \frac{x^k}{k!}  ,
\end{gather*}
is a generalized Laguerre polynomial.

\subsection*{Acknowledgements}

The authors are indebted to referees for their relevant  and  constructive comments and suggestions.
They  acknowledge f\/inancial support from the Universit\`a di Palermo.
S.T.A.~acknow\-led\-ges a~grant from the Natural Sciences and Engineering Research Council (NSERC) of Canada, F.B.~acknowledges support from GNFM, J.P.G.~thanks the CBPF and the CNPq for f\/inancial support and CBPF for hospitality.


\pdfbookmark[1]{References}{ref}
\LastPageEnding

\end{document}